\documentclass[runningheads, orivec]{llncs}
\usepackage[T1]{fontenc}

\usepackage[utf8]{inputenc}
\usepackage{graphicx}
\usepackage[margin=1.60in]{geometry}
\usepackage{amssymb}
\usepackage{tikz}
\usepackage[colorlinks=true, allcolors=blue!70]{hyperref}
\usetikzlibrary{decorations.pathmorphing}
\usetikzlibrary{decorations.markings}
\usetikzlibrary{calc}
\usetikzlibrary{arrows.meta}
\tikzset{city/.style={black, circle, draw, inner sep=1.2pt}}
\makeatletter
\newcommand\footnoteref[1]{\protected@xdef\@thefnmark{\ref{#1}}\@footnotemark}
\makeatother
\usepackage{placeins}
\usepackage{wrapfig}
\usepackage{subcaption}
\usepackage[most]{tcolorbox}
\newtcolorbox{problembox}{
  colback=white,
  colframe=black,
  boxrule=0.8pt,
  arc=2pt,
  left=6pt,
  right=6pt,
  top=6pt,
  bottom=6pt
}
\usepackage{enumitem}
\usepackage{thm-restate}
\usepackage{thmtools}
\usepackage{rotating}
\usepackage{comment}
\usepackage{soul}
\usepackage{varioref}
\usepackage{cleveref}
\usepackage{float}
\usepackage{algorithm}
\usepackage[toc,page]{appendix}
\usepackage{appendix}
\usepackage{changepage}
\usepackage[framemethod=TikZ]{mdframed}
\usepackage[noend]{algpseudocode}
\usepackage{xspace}
\usepackage[color=green!10!white,textsize=small]{todonotes}
\usepackage{multirow}
\usepackage{adjustbox}

\newcommand{\I}{\ensuremath{\mathcal{I}}\xspace}

\newcommand{\BC}{\textsc{Branching Cost}\xspace}
\newcommand{\BT}{\textsc{Branching Trajectory}\xspace}

\definecolor{red-brown}{rgb}{0.65, 0.16, 0.16}
\definecolor{darkspringgreen}{rgb}{0.09, 0.45, 0.27}
\colorlet{shadecolor}{gray!50}

\newcommand\journal[1]{}

\newcommand{\dmax}{\delta_{\text{max}}^{s\to s'}}
\newcommand{\dmin}{\delta_{\text{min}}^{s\to s'}}

\newcommand{\drefmax}{\delta^{s\to s'}_{\text{{\bf r}max}}}
\newcommand{\drefmin}{\delta^{s\to s'}_{\text{{\bf r}min}}}
\newcommand{\tsesi}{t-s'+s}

\graphicspath{{./figs/}}
\title{Optimal Trajectories in Discrete Space with Acceleration Constraints}

\author{Arnaud Casteigts\inst{1,2} \and
Matteo De Francesco\inst{1}
\and Pierre Leone\inst{1}}
\authorrunning{A. Casteigts and M. De Francesco and P. Leone} 

\institute{Department of Computer Science, University of Geneva, Switzerland \and
LaBRI, University of Bordeaux, France}

\begin{document}
\maketitle

\begin{abstract}
  In a recreational column of the Scientific American, Martin Gardner presented in 1973 a game called \textsc{Racetrack}, consisting of computing an optimal trajectory for a vehicle on a race circuit, subject to acceleration constraints in discrete space~$\mathbb{Z}^2$. 
  In this model, each step consists of changing the position of the vehicle by a vector in $\mathbb{Z}^2$, with the constraints that two consecutive vectors differ by at most one unit in each dimension.
  We investigate two problems related to this model in arbitrary dimension in open space (no obstacles), where a \emph{configuration} of the vehicle consists of its current position and the last-used vector (concretely, a value in $\mathbb{Z}^d \times \mathbb{Z}^d$). The two problems are the following. In \BC, two configurations are given and the goal is to compute the minimum number of moves (length of a trajectory) between the two configurations. \BT has the same input and asks for a description of the trajectory.

  We obtain various results. First, we revisit known approaches solving \BC in 2D, clarifying the analysis and showing that this problem can be solved in constant time in any fixed number of dimensions $d$ (more generally, in $O(d \log d)$ time). We also show that \BT can also be solved in constant time for any fixed $d$, despite the fact that the length of the trajectory is not constant. The main ingredient is to show that there always exists \emph{at least one} optimal trajectory that can be compactly represented using only $O(1)$ intermediate configurations, with monotonic evolution between them. Among other uses, the latter implies that computing an optimal trajectory that visits a sequence of $n$ points at prescribed velocities in 2D or 3D can be done in linear time in the number of points.
  
\end{abstract}

\keywords{Discrete geometry, Trajectory optimization, Motion planning, Acceleration constraints}

\newpage
\section{Introduction}
\label{sec:intro}

The problem of computing optimal trajectories under constraints has a long history. One of the first occurrences is the brachistochrone problem posed in 1696~\cite{sussmann_brachystochrone}, which seeks the curve of fastest descent between two points under gravity. Since then, trajectory optimization has found applications in many fields, including robotics, ballistics, autonomous vehicles, and aerospace engineering. It has also been addressed using diverse approaches. Continuous methods from control theory are perhaps the most representative (see, e.g.~\cite{fawadarminroutesoptimizatio,valli2024continuoustimeoptimalcontroltrajectory,gong_oc_applications}). In particular, the subfield of kinodynamics focuses on motion planning under physical constraints on velocity, acceleration, and force/torque (see, e.g.
~\cite{lee2024trajectorymanifoldoptimizationfast,ortizharo2024idbrrtsamplingbasedkinodynamicmotion,wahba2024kinodynamicmotionplanningteam,malik_robot_trajectory} or \cite{kinodynamic_planning} for a survey).
Optimization methods~\cite{degroot2024topologydrivenparalleltrajectoryoptimization,discretecontinuousoptim}, and more recently, reinforcement learning~\cite{Trauth_2024,9597689,McMahon_2022} (see~\cite{wangsurveymotion} for a survey),
are also commonly used in motion planning problems.

In a recreational column of the Scientific American~\cite{gardner1973sim}, Martin Gardner presented a paper-and-pencil game called \textsc{Racetrack} consisting of a race circuit drawn on a grid, where each player moves a car located at integer points. The goal is to finish the race with less moves than the opponents, each move corresponding to a vector in $\mathbb{Z}^2$ that differs by at most one unit from the previous move (in each dimension). While being very simple, this model already captures basic inertia constraints like the impossibility to brake or accelerate fast and to make tight turns at high speed. An appealing feature of this model is that it is entirely discrete and offers a finite number of choices to the vehicle in each step, making it amenable to algorithmic investigation.

For some time, Gardner's model remained limited to recreational mathematics. An early unpublished manuscript~\cite{schmid2005vector} investigates a BFS-like algorithm for finding the optimal trajectory between two points that does not hit an obstacle. Another blog post~\cite{blogernie} observes that the structural complexity of the problem ranges from \textsf{P} to \textsf{PSPACE}-complete (and possibly beyond) depending on how obstacles are specified.
Along the same lines, Holzer and McKenzie~\cite{holzer2010computational} study the impact of whether the vehicle is allowed to touch the edge of the track, showing that a non-touching variant is \textsf{NL}-complete.

Although fundamental, the above results on the impact of obstacle representation on the structural complexity are not directly relevant to our work, which focuses on designing fast algorithmic techniques in open space. In this respect, the more recent work in~\cite{bekos2018algorithms} is quite relevant.
The authors of that paper consider race tracks made of straight road segments separated by right-angle turns, making it possible to split a trajectory into different sections (one for each road segment) that join at low speed at the turns. The problem then reduces mostly to guessing the low-speed configurations used at the turns, and computing the number of configurations required in between every two consecutive turns. In particular, a strategy for computing such cost in constant time is outlined in~\cite{bekos2018algorithms}.

The elegance of Gardner's model makes it also relevant in open space. In~\cite{vectortsp}, a version of the \textsc{TSP} called \textsc{Vector TSP} is studied, where the vehicle must visit a given set of points in $\mathbb{Z}^2$ using as few moves as possible in Gardner's model.

\subsection{Model and problems}
\label{sec:model-problems}

In this paper, we consider a $d$-dimensional generalization of Gardner's model in open space. The state of the vehicle at time $t$ is given by a \emph{configuration} $c_t=(p_t,v_t)$, where $p_t \in \mathbb{Z}^d$ is the position of the vehicle and $v_t \in \mathbb{Z}^d$ is its velocity.
A \emph{trajectory} is a sequence of configurations $c_0,\dots, c_\ell$ such that $v_{i+1}=p_{i+1} - p_{i}$ and $v_{i+1} - v_i \in \{-1,0,1\}^d$ for all $0 \le i < \ell$. The length $\ell$ of a trajectory is the number of configurations in the sequence, minus one (i.e., the number of moves). A trajectory is called shortest (or optimal) if this length is minimum, regardless of the distance traveled.
Following standard assumptions, we assume that basic arithmetic operations such as $-,+,\times,/,\sqrt{\ }$ can be performed in constant time on integers and real numbers, as well as comparisons, and floor and ceiling operations.
We are interested in the following two problems:

\begin{problembox}
\textbf{Problem:} \textsc{Branching Cost}
\smallskip

\emph{Input:} Two configurations $c$ and $c'$, both in $\mathbb{Z}^d \times \mathbb{Z}^d$\\
\emph{Output:} Length of a shortest trajectory from $c$ to $c'$? (i.e. ``branching $c$ to $c'$'')
\end{problembox}

\begin{problembox}
\textbf{Problem:} \textsc{Branching Trajectory}
\smallskip

\emph{Input:} Two configurations $c$ and $c'$, both in $\mathbb{Z}^d \times \mathbb{Z}^d$\\
\emph{Output:} A compact description of a shortest trajectory, given for each dimension by a list of pairs $(a,m)$, where $a \in \{-1,0,1\}$ is the acceleration applied for $m$ consecutive moves.
\end{problembox}

The two problems are illustrated by Figure~\ref{fig:traj_example1} (in two dimensions). Here, two configurations $c$ and $c'$ are given. The goal of \BC is to return the optimal length of a trajectory (here, $5$) and the goal of \BT is to return a description of a trajectory realizing this cost.\vspace{-3pt}
\begin{figure}[h]
  \centering
  \begin{tikzpicture}[thick,scale=0.35, every node/.style={scale=0.3, circle, fill,  
      minimum size=1pt}]
    \path (0,0) coordinate (a) {};
    \path (1,2) node (b) {};
    \path (3,3) node (c) {};
    \path (6,3) node (d) {};
    \path (8,2) node (e) {};
    \path (10,2) node (f) {};
    \path (12,3) node (g) {};
    
    \tikzstyle{every node}=[]
    \draw (a) edge[black,-stealth] node[black,pos=1,above left]{$c$} (b);
    \draw (b) edge[red,dashed,-stealth] (c);
    \draw (c) edge[red,dashed,-stealth] (d);
    \draw (d) edge[red,dashed,-stealth] (e);
    \draw (e) edge[red,dashed,-stealth] (f);
    \draw (f) edge[red,-stealth] node[pos=1,above]{\color{black}$c'$} (g);

    \path (-.3,-.3) coordinate (bidon);

    \draw[shadecolor,step=1cm,opacity=.3]   (current bounding box.south west) grid (current bounding box.north east);
  \end{tikzpicture}
  \caption{\label{fig:traj_example1}A trajectory branching configuration $c=((1,2),(1,2))$ to configuration $c'=((12,3),(2,1))$. The solid vectors are imposed by the input. The dashed vectors are to be determined. The red vectors are counted in the length (here $5$).}
\end{figure}
\vspace{-10pt}

\subsection{Contributions}
\label{sec:contributions}

We obtain different results for the two problems. For \BC, we revisit and refine several aspects of the analysis from Bekos \textit{et al.}~\cite{bekos2018algorithms} (see also the conference version~\cite{bekos2016}, although we refer to the journal version below). First, we generalize it from two dimensions to arbitrarily many dimensions. We also present a set of reduction rules that isolates a core subproblem. This subproblem is analyzed in depth, resulting in explicit formulas that one can use to derive the final result directly. Some ingredients of this analysis were present in~\cite{bekos2018algorithms} and some were omitted; for instance, the quadratic formulas were not explicitly given and the integral nature (and parity) of the variables was not properly addressed, resulting in some equations of~\cite{bekos2018algorithms} yielding an incorrect result in special cases (we provide an example in Section~\ref{sec:two-point}). The correct quantities are characterized by our Theorem~\ref{lem:formula}. Another difference with~\cite{bekos2018algorithms} is that the subproblems here do not consist only of computing a minimum length or testing a candidate length in a dimension; instead, we compute a compact representation of \emph{all} the feasible lengths in a dimension, at once. This makes the subsequent combination of $d$ dimensions immediate, namely, a simple intersection of all one-dimensional results. The resulting algorithm computes the branching cost between two $d$-dimensional configurations in $O(d \log d)$-time; in particular, in $O(1)$ time for any constant~$d$. 

For \BT, a folklore approach for solving this problem is to find a shortest path in the \emph{configuration graph} $G=(V,A)$, where $V$ is the set of possible configurations and the arcs $A$ represent possible transitions between configurations (mentioned, for example, in~\cite{schmid2005vector}). Clearly, this approach is not optimal. In fact, we show that \BT can also be solved in constant time for any fixed $d$. The key property is that for any realizable length, in each dimension, there exists a trajectory of that length that behaves monotonously between $O(1)$ (namely, $4$) well chosen intermediate configurations, resulting in a compact representation of size $O(1)$ for this trajectory. Our analysis relies, again in part, on informal ingredients present in~\cite{bekos2018algorithms}, although the \BT problem was not explicitly addressed in that paper. Concretely, we prove the claims from~\cite{bekos2018algorithms} that for all realizable lengths, there exists in each dimension a trajectory that can be split in three periods, the first and last being monotonic and the middle one being almost monotonic in the sense that it needs at most a single unit change of speed, resulting in $4$ intermediate configurations (Theorem~\ref{lem:existence_k1k2}).
As each dimension can be processed independently (once the target length is fixed), this results in an algorithm that computes in $O(d\log d)$ time a $O(d)$-size representation of an optimal trajectory in arbitrarily many dimensions $d$.

\subsubsection{Organization of the document.}

The document is organized by problem: \BC is studied in Section~\ref{sec:two-point} and \BT in Section~\ref{sec:BT}. 
Due to space restrictions, some of the proofs are deferred to the appendices.

\section{Branching Cost}
\label{sec:two-point}

In this section, we focus on the \textsc{Branching Cost} problem.
We start by presenting a general strategy for reducing a $d$-dimensional instance to several 1-dimensional instances. This strategy consists of computing at once, for each dimension, the set of \emph{all} feasible lengths of a trajectory (as opposed to an optimal or candidate length only). We show how this approach trivializes the subsequent combination of arbitrarily many dimensions. Then, we proceed with the 1-dimensional analysis, exploiting a case reduction that simplifies the problem to the point that it can be solved using closed formulas that we give explicitly. The resulting algorithm runs in $O(d \log d)$ time and returns the minimum cost that is compatible with all the dimensions.

\subsection{From $d$ dimensions to one dimension}
\label{sec:d_to1d}

Clearly, the projection of a valid $d$-dimensional trajectory in any dimension is also a valid $1$-dimensional trajectory. The converse is not true. In fact, the optimal $d$-dimensional cost may not correspond to the optimal cost in any of the dimensions. Consider the example in Figure~\ref{fig:traj_example}.
\begin{figure}[h]
  \centering
  \begin{subfigure}[!htb]{0.25\linewidth}
  \begin{tikzpicture}[thick,scale=0.3, every node/.style={scale=0.3, circle, fill,  
      minimum size=1pt}]
    \path (0,0) coordinate (a) {};
    \path (1,3) node (b) {};
    \path (1,5) coordinate (c) {};
    \path (2,6) coordinate (d) {};
    \path (3,6) coordinate (e) {};
    \path (4,5) coordinate (f) {};
    \path (4,3) coordinate (g) {};
    \path (4,1) coordinate (h) {};
    \path (5,0) coordinate (v) {};
    \path (6,0) coordinate (w) {};
    \path (7,1) coordinate (x) {};
    \path (9,3) coordinate (y) {};
    \path (10,6) node (z) {};
    
    \tikzstyle{every node}=[]
    \draw (a) edge[black,-stealth] node[black,pos=1,above left]{$c$} (b);
    \draw (y) edge[red,-stealth] node[pos=1,above=2pt]{\color{black}$c'$} (z);

    \path (-.3,-.3) coordinate (bidon);

    \draw[shadecolor,step=1cm,opacity=.3]   (-.5,-1) grid (10.5,7);
  \end{tikzpicture}
\end{subfigure}
~~~~
  \begin{subfigure}[!htb]{0.25\linewidth}
  \tikzset{every loop/.style={min distance=10mm,in=60,out=120,looseness=20}}
  \begin{tikzpicture}[thick,scale=0.3, every node/.style={scale=0.3, circle, fill,  
      minimum size=1pt}]
    \path (0,0) coordinate (a) {};
    \path (1,0) node (b) {};
    \path (3,0) coordinate (c) {};
    \path (6,0) coordinate (d) {};
    \path (8,0) coordinate (dd) {};
    \path (9,0) coordinate (e) {};
    \path (10,0) node (f) {};
    
    \tikzstyle{every node}=[]
    \draw (a) edge[black,-stealth] node[black,pos=1,above]{$c$} (b);
    \draw (b) edge[red!40,dashed,-stealth] (c);
    \draw (c) edge[red!40,dashed,-stealth] (d);
    \draw (c) edge[red!40,dashed,-stealth] (dd);
    \draw (dd) edge[red!40,dashed,-stealth] (e);
    \draw (e) edge[red!70,red,-stealth] node[pos=1,above]{\color{black}$c'$}  (f);

    \path (-.3,-.3) coordinate (bidon);

    \draw[shadecolor,step=1cm,opacity=.3]   (-.5,-1) grid (10.5,2);
  \end{tikzpicture}
\end{subfigure}
~~~~
  \begin{subfigure}[!htb]{0.1\linewidth}
  \begin{tikzpicture}[thick,scale=0.3, every node/.style={scale=0.3, circle, fill,  
      minimum size=1pt}]
    \path (1,0) coordinate (a) {};
    \path (1,3) node (b) {};
    \path (1,3) coordinate (y) {};
    \path (1,6) node (z) {};
    
    \tikzstyle{every node}=[]
    \draw (a) edge[black,-stealth] node[black,pos=1,above left]{$c$} (b);
    \draw (y) edge[red,-stealth] node[pos=1,above=2pt]{\color{black}$c'$} (z);


    \draw[shadecolor,step=1cm,opacity=.3]   (0,-1) grid (2,7);
  \end{tikzpicture}
\end{subfigure}
~
  \begin{subfigure}[!htb]{0.25\linewidth}
  \begin{tikzpicture}[thick,scale=0.3, every node/.style={scale=0.3, circle, fill,  
      minimum size=1pt}]
    \path (0,0) coordinate (a) {};
    \path (1,3) node (b) {};
    \path (2,5) coordinate (c) {};
    \path (3,6) coordinate (d) {};
    \path (4,6) coordinate (e) {};
    \path (5,5) coordinate (f) {};
    \path (5,3) coordinate (g) {};
    \path (5,1) coordinate (h) {};
    \path (6,0) coordinate (v) {};
    \path (7,0) coordinate (w) {};
    \path (8,1) coordinate (x) {};
    \path (9,3) coordinate (y) {};
    \path (10,6) node (z) {};
    
    \tikzstyle{every node}=[]
    \draw (a) edge[black,-stealth] node[black,pos=1,above left]{$c$} (b);
    \draw (b) edge[red!70,dashed,-stealth] (c);
    \draw (c) edge[red!70,dashed,-stealth] (d);
    \draw (d) edge[red!70,dashed,-stealth] (e);
    \draw (e) edge[red!70,dashed,-stealth] (f);
    \draw (f) edge[red!70,dashed,-stealth] (g);
    \draw (g) edge[red!70,dashed,-stealth] (h);
    \draw (d) edge[red!70,dashed,-stealth] (e);
    \draw (h) edge[red!70,dashed,-stealth] (v);
    \draw (v) edge[red!70,dashed,-stealth] (w);
    \draw (w) edge[red!70,dashed,-stealth] (x);
    \draw (x) edge[red!70,dashed,-stealth] (y);
    \draw (y) edge[red,-stealth] node[pos=1,above=2pt]{\color{black}$c'$} (z);


    \draw[shadecolor,step=1cm,opacity=.3]   (-.5,-1) grid (10.5,7);
  \end{tikzpicture}
\end{subfigure}
  \caption{From left to right: A $2D$ instance; An optimal trajectory in the $x$-dimension (length~$5$); An optimal trajectory in the $y$-dimension (length $1$); An optimal trajectory in two dimensions (length $11$).}
  \label{fig:traj_example}
\end{figure}
The optimal cost in the $x$-dimension is $5$, and the optimal cost in the $y$-dimension is $1$, but the optimal cost in the two combined dimensions is $11$, which is the smallest value compatible with both dimensions.

The approach of~\cite{bekos2018algorithms} to solve this problem in $2D$ consists of computing the minimum length in each dimension, then examining whether the largest of these lengths (say, wlog, the one of the first dimension) is feasible in the other dimension (say, the second dimension). If not, compute the next feasible length in the second dimension and examine whether this length is feasible in the first dimension. If this is still not the case, the next feasible length in the first dimension is chosen as the answer. 

The reason why this strategy always works is not immediate, it follows from an observation we make below, namely, that the feasible lengths in each dimension can always be described by at most two contiguous intervals, the later of which is right-open (to infinity). Thus, in the worst case, the maximum opening value of the last interval (over all dimensions) must be feasible in all the other dimensions.

The advantage of computing at once the interval (or bi-interval) $\I_i$ for each dimension $i$ renders the combination of all dimensions much simpler. Namely, one can subsequently intersect all these intervals and the minimum value in this intersection is the smallest compatible value, i.e., the answer of the problem.
Beyond simplicity, this solution generalizes straightforwardly to any number of dimensions: whatever $d$, the minimum length of a trajectory in $d$ dimensions is $\min \{\cap_{i\in[d]}\, \I_i\}$.

Using standard intersection algorithms (see e.g.~\cite{marzullo1984}), the intersection of $d$ multi-intervals, each containing $O(1)$ intervals, can be performed in time $O(d \log d)$; namely, $O(d \log d)$ time to sort the endpoints of the intervals, then $O(d)$ time to scan these endpoints and produce the intersection.
Provided that $1$-dimensional instances can be solved in constant time (Section~\ref{sec:1d}), the overall algorithm has running time $O(d \log d)$.

\subsection{One dimension}
\label{sec:1d}

The strategy suggested above requires to solve a more general one-dimensional problem than the one in~\cite{bekos2018algorithms}; namely, computing all the feasible lengths of a trajectory between two $1$-dimensional configurations. Observe that in one dimension, a configuration is just a pair $c=(x,s)$ of integers, where $x$ is the position and $s$ the speed (positive or negative) of the vehicle. Given $c=(x,s)$ and $c'=(x',s')$, the goal is to compute all feasible lengths of a trajectory from $c$ (excluded) to $c'$ (included), represented by a multi-interval $\I(c,c')$.
We use the following notations. For a certain multi-interval $\mathcal{I}$ and integer $k$, we write $\mathcal{I} + k$ for $\mathcal{I}$ shifted by $k$ units. For example, if $\mathcal{I} = [2,5]\cup[8,\infty)$, then $\mathcal{I} + 2 = [4,7]\cup[10,\infty)$. We also use $\downarrow$$(s)=(|s|\cdot(|s|-1)/2)$ for the minimum braking distance from a certain speed $s$ to full stop, and $\uparrow$$(s)=(|s|\cdot(|s|+1)/2)$ for the minimum acceleration distance from zero speed to a certain speed $s$. 

\subsubsection{Case reductions.}
\label{sec:case-reduction}

We will show that $\I(c,c')$ can either be computed directly in constant time, or be reduced in constant time to a special case where $x=0, x'\ge x, s \ge 0,$ and $s' \ge 0$ (handled in Section~\ref{sec:main-case}), this case being sufficiently restricted to subsequently obtain explicit formulas. To start, observe that we can trivially reduce to the case that $x=0$ and $x' \ge x$:

\begin{lemma}[$x=0$]
	\label{lem:reduce1}
	For all $(x,s)$ and $(x',s')$,  $\I((x,s),(x',s')) = \I((0,s),(x'-x,s'))$ 
\end{lemma}
\begin{proof}
  The validity of a trajectory is invariant by translation (in particular, by $-x$).$\hfill\qed$
\end{proof}

\begin{lemma}[$x'\ge x$]
	\label{lem:reduce2}
	For all $(x,s)$ and $(x',s')$,  $\I((x,s),(x',s')) = \I((-x,-s),(-x',-s'))$ 
\end{lemma}
\begin{proof}
  The validity of a trajectory is invariant by symmetry with respect to the origin.$\hfill\qed$
\end{proof}

\noindent
By Lemma~\ref{lem:reduce1} and~\ref{lem:reduce2}, we assume in the rest of this section that $x=0$ and $x'\ge 0$. The values of $s$ and $s'$ then fall into five possible cases:

\begin{itemize}
\item \textbf{Case 0:} $s = s' = 0$
\item \textbf{Case 1:} $s \le 0$ and $s' \le 0$ (and not case 0)
\item \textbf{Case 2:} $s<0$ and $s'>0$
\item \textbf{Case 3:} $s>0$ and $s'<0$
\item \textbf{Case 4:} $s \ge 0$ and $s' \ge 0$ (and not case 0)
\end{itemize}

\begin{restatable}{lemma}{Cases}
  \label{lem:reduce-all}
  Cases 0 to 3 can be solved or reduced to Case 4 in constant time.
\end{restatable}
\begin{proof}
  In this proof, we slightly abuse the terms left and right to denote the two orientations of the position axis (where infinite left is $-\infty$ and infinite right is $+\infty$).
\begin{itemize}
\item \textbf{Case 0.} Let $c=(0,0)$ and $c'=(x',0)$, with $x'\ge 0$. By Lemma~4 in~\cite{vectortsp}, covering a distance $d$ starting and ending at zero-speed takes exactly $\lceil 2\sqrt{d} \rceil$ steps. Moreover, as zero-speed configurations can be repeated arbitrarily many times, $\mathcal{I}(c,c')=[\lceil 2\sqrt{x'} \rceil,\infty)$.
\item \textbf{Case 1.} Let $c=(0,s)$ and $c'=(x',s')$ with $x'\ge 0, s\le 0,$ and $s'\le 0$. First, observe that $0\in\mathcal{I}(c,c')$ if and only if $c=c'$. More generally, for a nonempty trajectory, the vehicle must decelerate to zero speed at some point $p_1\le 0$, then move right to a second zero-speed configuration at some point $p_2\ge x'$, and finally move left towards $x'$ (note that the initial deceleration and the final acceleration phases may be empty, as a particular case). Between $p_1$ and $p_2$, the situation is like in Case 0, where the cost grows monotonically with the distance to be covered. Thus, the overall cost is minimized if $p_1=-\downarrow(s)$ is the first position where the vehicle can stop from $c$ and $p_2=x'+\uparrow(s')$ is the closest position from $x'$ where the vehicle can start at zero speed and reach $x'$ at speed $s'$. It follows that (apart from adding potentially $\{0\}$ to it), $\mathcal{I}(c,c')= |s| + \mathcal{I}((p_1,0),(p_2,0)) + |s'|$, where (by Lemma~\ref{lem:reduce1}) $\mathcal{I}((p_1,0),(p_2,0))=\mathcal{I}((0,0),(p_2-p_1,0))$, which is Case 0.
\item \textbf{Case 2.} Let $c=(0,s)$ and $c'=(x',s')$ with $x'\ge 0, s<0,$ and $s'>0$. In this case, the vehicle must make a turn on the left side, and come back to $x'$. 
  Let $p_1= -\downarrow$$(s)$ be the first position where the vehicle can stop from $c$. Let $p_2=x'-\uparrow$$(s')$ be the closest position from $x'$ where the vehicle can start at zero speed and reach $x'$ at speed $s'$. If $p_1 \le p_2$, then the vehicle has sufficient time from $p_1$ to reach $x'$ at speed $s'$, thus the best trajectory is to turn at $p_1$. Otherwise, the vehicle must turn at $p_2 < p_1$. In the first case, this gives us $\mathcal{I}(c,c')= |s| + \mathcal{I}((p_1,0),(x',s'))$, where (by Lemma~\ref{lem:reduce1}) $\mathcal{I}((p_1,0),(x',s'))=\mathcal{I}((0,0),(x'-p_1,s'))$, which is in Case 4. In the second case, we have $\mathcal{I}(c,c')=\mathcal{I}((0,s),(p_2,0)) + s'$, where (by Lemma~\ref{lem:reduce2}) $\mathcal{I}((0,s),(p_2,0)) = \mathcal{I}((0,-s),(-p_2,0))$, which is in Case 4.
\item \textbf{Case 3.} Let $c=(0,s)$ and $c'=(x',s')$ with $x'\ge 0, s>0,$ and $s'<0$. The arguments are similar to the ones of Case 2, except that the two reference points are on the right side of~$x'$, namely $p_1=\downarrow$$(s)$ and $p_2=x'+\uparrow$$(s')$. Here, if $p_1 \ge p_2$, we have $\mathcal{I}(c,c')=s + \mathcal{I}((p_1,0),(x',s'))$, where (by Lemma~\ref{lem:reduce2} and Lemma~\ref{lem:reduce1}, in this order) $\mathcal{I}((p_1,0),(x',s')) = \mathcal{I}((-p_1,0),(-x',-s'))= \mathcal{I}((0,0),(-x'+p_1,-s'))$, which is in Case 4. If $p_2 > p_1$, we have $\mathcal{I}(c,c')=\mathcal{I}((0,s),(p_2,0)) + |s'|$, where (by Lemma~\ref{lem:reduce1}) $\mathcal{I}((0,s),(p_2,0))=\mathcal{I}((0,s),(p_2-x,0))$, which is again in Case 4.
\end{itemize}
\medskip

\noindent
All the operations involved in the above cases can be performed in constant time. The only operations for which this is not obviously true are interval-shifting operations of the form $\mathcal{I}+k$, where the running time depends on the number of intervals in $\mathcal{I}$. However, it will become clear from Algorithm~\ref{algo:minimum-t} in~\cref{sec:main-case} that Case~4 produces multi-intervals that contain at most $3=O(1)$ characteristic values.$\hfill\qed$
\end{proof}

\subsubsection{The remaining case.}
\label{sec:main-case}

Let $c=(x,s)$ and $c'=(x',s')$ be two configurations such that $x = 0$, $x' \ge 0$, $s \ge 0$, and $s' \ge 0$. 
Let $\delta=x'$ be the distance to be traveled.
Recall that the goal is to determine the feasible lengths of a trajectory $c_1,\dots,c_\ell$ such that $c_1$ is a valid successor of $c$ and $c_\ell=c'$.

The difficulty of this case is that a turn may or may not be needed, resulting in feasible lengths described by one or two intervals. To build some intuition, consider a simple example where $c=(0,6)$ and $c'=(24,5)$ (see Figure~\ref{fig:feasible}). Due to $c$, the speed of $c_1$ can be either $5, 6,$ or $7$, the speed of $c_\ell=c'$ must be exactly $5$, and the distance to be covered is $\delta=24$. The shortest trajectory has length $4$, corresponding to the sequence of speeds $7+6+6+5$ ($=24$). There is also a solution of length $5$, corresponding to speeds $5+4+5+5+5$, or length $6$ using $5+4+3+3+4+5$. However, a trajectory of length $7$ is impossible, as any sum of seven terms respecting the model would cover a distance larger than $24$. In fact, the next feasible number of steps is $16$ steps, which corresponds to the vehicle decelerating, stopping, moving backward at negative speed, and coming back to $c'$ at speed $s'$, using the sequence $5+4+3+2+1+0-1-2-2-1+0+1+2+3+4+5$ ($=24$). As this trajectory contains a zero-speed configuration, this configuration can be repeated, so all lengths above $16$ are also feasible and the feasible lengths are $\mathcal{I}(c,c')=[4,6]\cup [16,\infty)$. 
\vspace{-5pt}

\begin{figure}[h]
	\centering
\begin{tikzpicture}[xscale=.5,yscale=.14]
	\tikzstyle{every node}=[draw,circle,inner sep=.8pt]
	\path (1,5) node (min1){};
	\path (2,10) node (min2){};
	\path (3,14) node (min3){};
	\path (4,18) node (min4){};
	\path (5,21) node (min5){};
	\path (6,24) node (min6){};
	\path (7,26) node (min7){};
	\path (8,28) node (min8){};
	\path (9,29) node (min9){};
	\path (10,30) node (min10){};
	\path (11,30) node (min11){};
	\path (12,30) node (min12){};
	
	\path (13,29) node (min13){};
	\path (14,28) node (min14){};
	\path (15,26) node (min15){};
	\path (16,24) node (min16){};
	\path (22,2) coordinate (mininf){};

	\path (1,5) node (max1){};
	\path (2,11) node (max2){};
	\path (3,18) node (max3){};
	\path (4,25) node (max4){};
	\path (5,33) node (max5){};
	\path (5.4,37) coordinate (maxinf){};

	\tikzstyle{every node}=[draw,circle,fill=darkgray,inner sep=.8pt]
	\path (4,24) node (ok4){};
	\path (5,24) node (ok5){};
	\path (6,24) node (ok6){};
	\path (16,24) node (ok16){};
	
	\tikzstyle{every node}=[]
	\draw[->] (0,0) -- (18,0) node[right] {{\large Length $t$}};
	\draw[->] (0,0) -- (0,37) node[above] {Distance $\delta$};
	\draw[dotted] (0,24) -- (20,24);
	\path (0,24) node[left] {$24$};
	\path (1,-.5) node[below] (lab1){$t_{min}=1$};
	\path (4,-.5) node[below] (lab4){$4$};
	\path (6,-.5) node[below] (lab6){$6$};
	\path (16,-.5) node[below] (lab16){$16$};

	\draw[thin,gray] (1,5) -- (lab1);
	\draw[thin,gray] (ok4) -- (lab4);
	\draw[thin,gray] (ok6) -- (lab6);
	\draw[thin,gray] (ok16) -- (lab16);

	\draw[blue] (min1)--(min2)--(min3)--(min4)--(min5)--(min6)--(min7)--(min8)--(min9)--(min10)--(min11)--(min12)--(min13)--(min14)--(min15)--(min16);
	\draw[blue,dashed] (mininf) edge[bend right=3] (min16) node[right]{$\delta_{min}^{6\to 5}$};
	\draw[red] (max1)--(max2)--(max3)--(max4)--(max5);
	\draw[red,dashed] (max5)--(maxinf) node[above]{$\delta_{max}^{6\to 5}$};

	\tikzstyle{every node}=[darkgray,below right=-2pt,font=\tiny]
	\path (2,10) node{$5$+$5$};
	\path (3,14) node{$5$+$4$+$5$};
	\path (4,18) node{$5$+$4$+$4$+$5$};
	\path (5,21) node{$5$+$4$+$3$+$4$+$5$};
	\path (6,24) node[black]{$5$+$4$+$3$+$3$+$4$+$5$};
	\draw (min11) -- (11.2,31.5) node[pos=1,above right=-2pt]{$5$+$4$+$3$+$2$+$1$+$0$+$1$+$2$+$3$+$4$+$5$};
	\path (16,24) node[black,above right=-1pt]{$5$+$4$+$3$+$2$+$1$+$0$-$1$-$2$-$2$-$1$+$0$+$1$+$2$+$3$+$4$+$5$};

	\tikzstyle{every node}=[darkgray,left,font=\tiny]
	\path (1,5) node{$5$};
	\path (2,11) node{$6$+$5$};
	\path (3,18) node{$7$+$6$+$5$};
	\path (4,25) node[above left=-2pt]{$7$+$7$+$6$+$5$};
	\path (5,33) node{$7$+$8$+$7$+$6$+$5$};
	\draw (ok4)+(-1,-1) -- (ok4) node[pos=0,black,below left=-2pt]{$7$+$6$+$6$+$5$};
	\draw (ok5)+(.2,2) -- (ok5) node[xshift=8pt,pos=0,black,above]{$5$+$4$+$5$+$5$+$5$};
	
\end{tikzpicture}
\caption{\label{fig:feasible} Feasible numbers of steps from $(0,6)$ to $(24,5)$.}
\end{figure}
\vspace{-5pt}

Let us take a more systematic look at what happens.
First observe that the vehicle must change its speed from $s$ to $s'$, so $t_{min}:=| s' - s |$ is an obvious lower bound on any feasible length.
For all $t \ge t_{min}$, one can define two functions $\delta_{max}^{s\to s'}(t)$ and $\delta_{min}^{s\to s'}(t)$ that give respectively the maximum and minimum distances the vehicle could cover in $t$ steps while starting at speed $s$ and ending at speed~$s'$. These two functions are depicted in red and blue in Figure~\ref{fig:feasible}. The fact that $x' \ge x$ and both $s$ and $s'$ are nonnegative (due to the case reduction) imply that the functions $\dmin(t)$ and $\dmax(t)$ are always of the depicted form. Namely, $\dmax$ is always increasing, while $\dmin$ first increases and then decreases. In addition, $\dmin$ is below $\dmax$ for all $t\ge t_{min}$, and both coincide at $t=t_{min}$.
More precisely, both functions can be expressed through the following explicit equations (the corresponding calculation is given Appendix~\ref{sec:appendix_computation}).

    \begin{restatable}{theorem}{Formula}
      \label{lem:formula}
      $\dmin$ and $\dmax$ are characterized by the following equations:~\smallskip\\
	\indent
	$\dmin(t) = \bigl\lceil \dfrac{1}{4} \left(3s-t+s'\right) \left(\tsesi\right) + \alpha\bigr\rceil$ and\\
        \indent
	$\dmax(t) = \bigl\lfloor\dfrac{1}{4} \left(3s'+t+s\right) \left(\tsesi\right) + \alpha\bigr\rfloor$, \hfill where $\alpha = \dfrac{1}{2}(s'+s+1)\cdot \left(s'-s\right)$.
\end{restatable}

It is also not difficult to see that $\dmin(s+s')$ is a maximum value for $\dmin$, a fact used in the algorithm. (This is true before rounding, and the rounding is monotonic.)
Observe that the equations given in~\cite{bekos2018algorithms} do not correctly capture this exact calculation. For example, if we take $s=s'=2$ and $t=6$, then Equations~5 and~7 in~\cite{bekos2018algorithms} yield $\dmin(6)=2$ whereas the correct value is $3$.

The above point of view suggests a natural geometric interpretation of the feasible lengths of a trajectory. Namely, the solutions are exactly the integer values $t\geq |s'-s|$ such that $\dmin(t) \leq \delta \leq \dmax(t)$ (bold dots in Figure~\ref{fig:feasible}). These inequalities are indeed sufficient (if the figure is not convincing enough, Theorem~\ref{lem:existence_k1k2} in~Section~\ref{sec:BT} establishes that one can construct a trajectory for any integer distance between the two extrema).
The resulting strategy is illustrated in~\Cref{fig:four-types} (and formally given by~\Cref{algo:minimum-t}). The exact calculation depends on how the line $y=\delta$ interacts with $\dmin$ and $\dmax$. If $\delta$ is always above or equal to $\dmin$ (Type A), then the solutions are the values of $t$ in the interval $[t_1,\infty)$, where $t_1$ is the first integer such that $\delta_{max}^{s\to s'}(t_1)\ge \delta$. As $\dmax$ is a quadratic function, the equation $\dmax(t) = \delta$ has two solutions, of which we select the second.
\begin{figure}[h]
  \centering
  \quad\quad
\begin{tikzpicture}[xscale=.6,yscale=.5]
	\draw (0.4,0) -- (13.5,0);
	\draw (0.4,0) -- (0.4,7);
	\draw[red] (1,1.5) edge[bend right=12] (3.5,7);
	\draw[blue] (1,1.5) edge[bend left=18] (5,4.5);
	\draw[blue] (5,4.5) edge[bend left=20] (9.5,3.8);
	\draw[blue] (9.5,3.8) edge[bend left=15] (13,0);

	\draw[fill=gray!20,draw=none] (1,1.5) to[bend left=18] (5,4.5) to[bend left=20] (9.5,3.8) to[bend left=15] (13,0) -- (13.5,0) -- (13.5,7) -- (3.51,7) to[bend left=12] (1.01,1.5);
	\draw[blue] (1,1.5) edge[bend left=18] (5,4.5);
	\draw[blue] (5,4.5) edge[bend left=20] (9.5,3.8);
	\draw[blue] (9.5,3.8) edge[bend left=15] (13,0);

	\draw[dashed] (0.4,1) -- (14,1) node[right]{Type D: $[t_3,\infty)$};
	\draw[dashed] (0.4,2.7) -- (14,2.7) node[right]{Type C: $[t_3,\infty)$};
	\draw[dashed] (0.4,4) -- (14,4) node[right]{Type B: $[t_1,t_2]\cup [t_3,\infty)$};
	\draw[dashed] (0.4,5.5) -- (14,5.5) node[right]{Type A: $[t_1,\infty)$};

	\tikzstyle{every node}=[draw,circle,fill=darkgray,inner sep=.8pt]
	\path (12,1) node[fill=white] {};
	\path (12.6,1) node[label=80:$t_3$] {};
 
	\path (1.6,2.7) node[fill=white] {};
	\path (2.2,2.7) node[fill=white] {};
	\path (10.8,2.7) node[fill=white] {};
	\path (11.4,2.7) node[label=80:$t_3$] {};
	
	\path (2.2,4) node[fill=white] {};
	\path (2.8,4) node[pin={[pin distance=5pt]60:$t_1$}] {};
	\path (3.4,4) node[pin={[pin distance=5pt]60:$t_2$}] {};
	\path (4,4) node[fill=white] {};
	\path (8.8,4) node[fill=white] {};
	\path (9.4,4) node[label=80:$t_3$] {};

	\path (2.8,5.5) node[fill=white] {};
	\path (3.4,5.5) node[label=80:$t_1$] {};
      \end{tikzpicture}\medskip
      
      \caption{\label{fig:four-types}Four types of interactions between $\delta$, $\dmin(t)$, and $\dmax(t)$.}
    \end{figure}
    If $\delta$ is small enough that $\delta < \dmin(t_{min})$ (Type~D), then the solutions are the values of $t$ in the interval $[t_3,\infty)$, where $t_3$ is the first integer such that $\delta_{min}^{s\to s'}(t_3)\le \delta$. If none of these two cases apply, then the line $y=\delta$ must have two disjoint segments in the region between $\dmin$ and $\dmax$.
However, the first intersection may or may not contain integer points. This can be detected using the value $t_2$, defined as the last integer $x$-coordinate before the line crosses $\dmin$ for the first time. If $t_2$ is larger or equal to $t_1$, then there is at least one integer value in this segment, so the feasible lengths are $[t_1,t_2]\cup [t_3,\infty)$ (possibly, with $t_2=t_1$). Otherwise, the first intersection can be ignored, and the solutions are the same as in type D, namely $[t_3, \infty)$.

Now, let us provide the explicit formulas for calculating the points $t_1, t_2,$ and $t_3$ as a function of $\delta, s$, and $s'$. 
These formulas follow from the discriminants of the equations $\dmin$ and $\dmax$, respectively $\mathrm{\Delta}_\text{min} = 2(s')^2 + 2s^2 + 2s' - 2s - 4\delta$ and $\mathrm{\Delta}_\text{max} = 2(s')^2 + 2s^2 - 2s' + 2s + 4\delta$, namely:
$$ t_1= \lceil -s'-s+\sqrt{\mathrm{\Delta}_\text{max}} \rceil,\, t_2= \lfloor s'+s-\sqrt{\mathrm{\Delta}_\text{min}} \rfloor,\, t_3= \lceil s'+s+\sqrt{\mathrm{\Delta}_\text{min}} \rceil.$$

\noindent
The complete procedure corresponding to this strategy is given in Algorithm~\ref{algo:minimum-t}.
\begin{algorithm}[h]
	\caption{Feasible times for branching two configurations in one dimension.}
	\label{algo:minimum-t}
	\begin{algorithmic} 
			\State {\bf Input}: two configurations $c = (0,s)$ and $c'=(x',s')$
			\State {\bf Output} feasible lengths $\mathcal{I}(c,c')$ of a branching trajectory from $c$ to $c'$.
          \State $\delta \gets x'$
          \State $t_1 \gets \lceil -s'-s+\sqrt{\mathrm{\Delta}_\text{max}} \rceil$
          \If{$\delta \ge \dmin(s + s')$} \Comment{Type A}\label{line:A}
          \State \Return $[t_1, \infty)$\label{line:AA}
          \Else
          \State $t_3 \gets \lceil s'+s+\sqrt{\mathrm{\Delta}_\text{min}} \rceil$
          \State $t_{min} \gets |s'-s|$ 
          \If{$\delta < \dmin(t_{min})$} \Comment{Type D}\label{line:D}
          \State \Return $[t_3,\infty)$\label{line:DD}
          \Else
          \State $t_2 \gets \lfloor s'+s-\sqrt{\mathrm{\Delta}_\text{min}} \rfloor$
          \If{$t_2 \ge t_1$} \Comment{Type B}\label{line:B}
          \State \Return $[t_1,t_2]\cup[t_3,\infty)$\label{line:BB}
          \Else \Comment{Type C}\label{line:C}
          \State \Return $[t_3,\infty)$\label{line:CC}
          \EndIf\
          \EndIf\
          \EndIf\
	\end{algorithmic}
\end{algorithm}

\begin{lemma}
  \label{lem:case4-constant}
  Algorithm~\ref{algo:minimum-t} runs in constant time.
\end{lemma}
\begin{proof}
Both $\dmin$ and $\dmax$ are not quadratic functions because of the rounding to ceiling and floor. However, since $t$ and $\delta$ are both integers value, equations of the form $\dmin(t) - \delta = 0$ and $\dmax(t) - \delta = 0$ can be solved using the unrounded quadratic expressions. Therefore, removing the rounding preserves the feasibility inequalities. 
Computing the resulting roots and rounding them gives the interval endpoints, thus $t_1, t_2,$ and $t_3$ can be determined in constant time by extracting the first and second roots.$\hfill\qed$
\end{proof}

\section{Branching Trajectory}
\label{sec:BT}

As already explained, \BT can be solved using a shortest path algorithm in the configuration graph, although this is not efficient. In this section, we show that this problem can be solved in constant time for any fixed $d$, despite the fact that the trajectory has a non-constant length. Namely, one can first use \BC to compute the optimal length $\ell$ of a trajectory. Then, knowing $\ell$, we show that there exists, for each dimension, a trajectory of length $\ell$ that can be fully characterized using at most $O(1)$ well-chosen intermediate configurations, each computable in constant time, with monotonic evolution in between. The resulting $d$-dimensional trajectory is then the direct combination of $d$ trajectories in 1D, which means it can be described using $O(d)$ one-dimensional configurations overall.

As the problem is independent in each dimension once the target length $\ell$ is known, it is sufficient to focus on the one dimensional case: given two configurations $c=(x,s)$, $c'=(x',s')$ and a feasible length $\ell$, compute the intermediate configurations and characterize the monotonic evolution between them.
To start, observe that we can apply the same case reduction procedure as the one exposed in Section~\ref{sec:case-reduction}. The key observation here is that these reductions involve only translations (Lemma~\ref{lem:reduce1}), symmetries w.r.t. the origin (Lemma~\ref{lem:reduce2}), and monotone segments of deceleration or acceleration (see the reduced forms of Cases 0 to 3 in Lemma~\ref{lem:reduce-all}). All of these operations can be interpreted without ambiguity at the level of the trajectories computed in this section. Thus, we assume $x=0, x'\ge 0, s\ge 0,$ and $s'\ge 0$.

\subsection{New distance functions and construction}
Let $\delta=x'$ be the distance to be covered and let $\ell\geq |s'-s|$ be an integer satisfying
$\dmin(\ell) \leq \delta \leq \dmax(\ell)$. We construct a trajectory of length $\ell$ from $c$ to $c'$, proving that these inequalities suffice for feasibility.
In addition to $\dmin$ and $\dmax$, we define two other functions $\drefmin$ and $\drefmax$ satisfying $\dmin(\ell)\le \drefmin(\ell)\le \drefmax(\ell)\le \dmax(\ell)$ (see Figure~\ref{fig:distancesregions}). Precisely: 
\begin{itemize}
\item When $s < s'$, the sequence of configurations for $\drefmin$ maintains the same constant speed $s$ for $\ell-|s'-s|$ times, 
and then constantly increase its speed by 1 from $s$ to $s'$ for $|s'-s|$ times; instead for $\drefmax$, the sequence of 
configurations is composed by a constant increases of the speed from $s$ to $s'$ for $|s'-s|$ times, followed by keeping the 
same speed $s'$ for $\ell-|s'-s|$ times. Denoting $\alpha_1 = (s'-s)(s+s'+1)/2$, both distances admit a 
closed form equation and a $O(1)$-size representation:
\begin{align*}
	\drefmin &= s(\ell-|s'-s|) + \alpha_1 && (0,\ell-s'+s)(+,s'-s) \\
	\drefmax &= s'(\ell-|s'-s|) + \alpha_1 && (+,s'-s)(0,\ell-s'+s)
\end{align*}
where $(0,s)/(+,s)/(-,s)$ means, respectively, keeping the same speed/accelerating/decelerating for $s$ successive configurations. We show that such representation is possible for all distances.
 
\item When $s \geq s'$, the sequence of configurations for $\drefmin$ constantly decrease its speed from $s$ to $s'$ for $|s'-s|$ times, 
and then keep the same speed $s'$ for $\ell-|s'-s|$ times; instead for $\drefmax$, the sequence of configurations maintain the 
same speed $s$ for $\ell-|s'-s|$ times, and then constantly decrease its speed by 1 from $s$ to $s'$ for $|s'-s|$ times. Denoting $\alpha_2 = (s-s')(s+s'-1)/2$, both distances admit a closed form equation and a $O(1)$-size representation	
\begin{align*}
	\drefmin = s'(\ell-|s'-s|) + \alpha_2 && (-,s-s')(0,\ell-s+s') \\
	\drefmax = s(\ell-|s'-s|) + \alpha_2 && (0,\ell-s+s')(-,s-s')
\end{align*}
\end{itemize}
\vspace{-7pt}

\begin{figure}[!htb]
	\centering
	\begin{subfigure}[t]{0.49\linewidth}
		\begin{tikzpicture}[thick, scale=0.5, every node/.style={scale=1}]
			\coordinate (s) at (0,1);
			\coordinate (e) at (5,2);
			\coordinate (dmin) at (2,-1);
			\coordinate (dmax) at (3,4); 

			\draw (-0.15,2) node[left]{$s'$}--(0.15,2);
			\draw (-0.15,1) node[left]{$s$}--(0.15,1);
			\draw (5,-0.15) node[below]{$\ell$}--(5,0.15);

			\draw[-stealth](0,-1) -- (0,4) node[thick, above] {speed};
			\draw[-stealth](0,0) -- (6.5,0) node[thick, above] {length};

			\draw (s) -- (dmax) -- (e);
			\draw (s) -- (dmin) -- (e);
			
			\draw [ultra thick, draw=none, fill=blue, opacity=0.1] (s) -- (4,1) -- (e) -- (dmin) -- cycle;

			\draw [ultra thick, draw=none, fill=brown, opacity=0.1] (s) -- (1,2) -- (e) -- (4,1) -- cycle;

			\draw [ultra thick, draw=none, fill=red, opacity=0.1] (1,2) -- (dmax) -- (e) -- cycle;

			\node [blue] at (2,0.5) {\bf A};
			\node [brown] at (2.5,1.5) {\bf B};
			\node [red] at (3,2.5) {\bf C};

			
			\draw [red, line width=0.4mm] (1,2) -- (dmax) -- (e);
			\draw [brown, line width=0.4mm] (1,2) -- (e);
			\draw [darkspringgreen, line width=0.4mm] (s) -- (4,1);
			\draw [blue, line width=0.4mm] (s) -- (dmin) -- (4,1);

			\draw [brown, line width=0.4mm, dashed,dash pattern=on 2pt off 2pt,dash phase=2pt] (0,1) -- (1,2);
			\draw [red, line width=0.4mm, dashed,dash pattern=on 2pt off 2pt] (0,1) -- (1,2);

			\draw [darkspringgreen, line width=0.4mm, dashed,dash pattern=on 2pt off 2pt,dash phase=2pt] (4,1) -- (e);
			\draw [blue, line width=0.4mm, dashed, dash pattern=on 2pt off 2pt] (4,1) -- (e);

			\tikzstyle{every node}=[draw,circle,fill=white,inner sep=.8pt]
			\path (s) node{};
			\path (0.5,1.5) node{};
			\path (1,2) node{};
			
			\path (1.5,2.5) node{};
			\path (2,3) node{};
			\path (2.5,3.5) node{};
			\path (3,4) node{};
			\path (3.5,3.5) node{};
			\path (4,3) node{};
			\path (4.5,2.5) node{};

			\path (1.5,2) node{};
			\path (2,2) node{};
			\path (2.5,2) node{};
			\path (3,2) node{};
			\path (3.5,2) node{};
			\path (4,2) node{};
			\path (4.5,2) node{};

			\path (0.5,1) node{};
			\path (1,1) node{};
			\path (1.5,1) node{};
			\path (2,1) node{};
			\path (2.5,1) node{};
			\path (3,1) node{};
			\path (3.5,1) node{};
			\path (4,1) node{};

			\path (0.5,0.5) node{};
			\path (1,0) node{};
			\path (1.5,-0.5) node{};
			\path (dmin) node{};
			\path (2.5,-0.5) node{};
			\path (3,0) node{};
			\path (3.5,0.5) node{};
			\path (4.5,1.5) node{};

			\path (5,2) node{};

		\end{tikzpicture}
	\end{subfigure}
	\begin{subfigure}[t]{0.49\linewidth}
		\begin{tikzpicture}[thick, scale=0.5, every node/.style={scale=1}]
			\coordinate (s) at (0,2);
			\coordinate (e) at (5,1);
			\coordinate (dmin) at (3,-1);
			\coordinate (dmax) at (2,4); 

			\draw (-0.15,2) node[left]{$s$}--(0.15,2);
			\draw (-0.15,1) node[left]{$s'$}--(0.15,1);
			\draw (5,-0.15) node[below]{$\ell$}--(5,0.15);

			\draw[-stealth](0,-1) -- (0,4) node[thick, above] {speed};
			\draw[-stealth](0,0) -- (6.5,0) node[thick, above] {length};

			\draw (s) -- (dmax) -- (e);
			\draw (s) -- (dmin) -- (e);
			
			\draw [ultra thick, draw=none, fill=blue, opacity=0.1] (s) -- (1,1) -- (e) -- (dmin) -- cycle;

			\draw [ultra thick, draw=none, fill=brown, opacity=0.1] (s) -- (4,2) -- (e) -- (1,1) -- cycle;

			\draw [ultra thick, draw=none, fill=red, opacity=0.1] (s) -- (dmax) -- (4,2) -- cycle;

			\node [blue] at (3,0.5) {\bf A};
			\node [brown] at (2.5,1.5) {\bf B};
			\node [red] at (2,2.5) {\bf C};
			
			\draw [red, line width=0.4mm] (s) -- (dmax) -- (4,2);
			\draw [brown, line width=0.4mm] (s) -- (4,2);
			\draw [darkspringgreen, line width=0.4mm] (1,1) -- (e);
			\draw [blue, line width=0.4mm] (1,1) -- (dmin) -- (e);

			\draw [brown, line width=0.4mm, dashed,dash pattern=on 2pt off 2pt,dash phase=2pt] (4,2) -- (e);
			\draw [red, line width=0.4mm, dashed,dash pattern=on 2pt off 2pt] (4,2) -- (e);

			\draw [darkspringgreen, line width=0.4mm, dashed,dash pattern=on 2pt off 2pt,dash phase=2pt] (s) -- (1,1);
			\draw [blue, line width=0.4mm, dashed, dash pattern=on 2pt off 2pt] (s) -- (1,1);

			\tikzstyle{every node}=[draw,circle,fill=white,inner sep=.8pt]
			\path (0,2) node{};
			\path (0.5,1.5) node{};
			\path (1,1) node{};
			
			\path (0.5,2.5) node{};
			\path (1,3) node{};
			\path (1.5,3.5) node{};
			\path (2,4) node{};
			\path (2.5,3.5) node{};
			\path (3,3) node{};
			\path (3.5,2.5) node{};
			\path (4.5,1.5) node{};

			\path (0.5,2) node{};
			\path (1,2) node{};
			\path (1.5,2) node{};
			\path (2,2) node{};
			\path (2.5,2) node{};
			\path (3,2) node{};
			\path (3.5,2) node{};
			\path (4,2) node{};

			\path (1.5,1) node{};
			\path (2,1) node{};
			\path (2.5,1) node{};
			\path (3,1) node{};
			\path (3.5,1) node{};
			\path (4,1) node{};
			\path (4.5,1) node{};

			\path (1.5,0.5) node{};
			\path (2,0) node{};
			\path (2.5,-0.5) node{};
			\path (dmin) node{};
			\path (3.5,-0.5) node{};
			\path (4,0) node{};
			\path (4.5,0.5) node{};

			\path (5,1) node{};

		\end{tikzpicture}
	\end{subfigure}
	\caption{The three distance regions {\bf A}, {\bf B}, {\bf C} defined by the four distance functions $\dmin(\ell)$ (blue solid line), $\drefmin(\ell)$ (dark green solid line), $\drefmax(\ell)$ (brown solid line), $\dmax(\ell)$ (red solid line) for two given configuration from $(0,s)$ to $(x',s')$. 
	Each distance function corresponds to the area under the corresponding curve.}
	\label{fig:distancesregions}
      \end{figure}
      \vspace{-5pt}

Using the same graphical representation as \cite{bekos2018algorithms}, we observe how combining such additional distances ($\drefmin$, $\drefmax$) with the already defined $\dmin$ and $\dmax$, we induce a partition of the \emph{possible distances} 
(Figure \ref{fig:distancesregions}). Therefore, we can distinguish immediately whether a trajectory of distance $\delta$ belongs to a specific region, namely \textbf{A}, \textbf{B}, or \textbf{C}, depending on whether $\delta$ is in the range $\mathcal{D}_\text{A} = \left[\dmin(\ell),\drefmin(\ell)\right]$, $\mathcal{D}_\text{B} = \left[\drefmin(\ell),\drefmax(\ell)\right]$, or $\mathcal{D}_\text{C} = \left[\drefmax(\ell),\dmax(\ell)\right]$. 

This yields three cases to discuss if $s\ge s'$ and three more if $s<s'$. We will see for each region (\textbf{A}, \textbf{B}, \textbf{C}) how the trajectory is computed.
Let us illustrate how a constant-size representation is obtained by assuming $\delta \in \mathcal{D}_A$ and $s\le s'$. The case $s > s'$ is handled separately in Appendix~\ref{sec:appendix_existence}. Let us denote the quantity 
$D = \drefmin(\ell)-\delta \in \mathbb{Z}_{\geq0}$. 

\begin{figure}[!htb]
	\centering
	\begin{subfigure}[!t]{0.3\textwidth}
		\centering
		\begin{tikzpicture}[thick, scale=0.4, every node/.style={scale=0.6}]
			\coordinate (s) at (0,2);
			\coordinate (e) at (7,3);
			\coordinate (dmin) at (3,-1);

			\draw (-0.15,3) node[left]{$s'$}--(0.15,3);
			\draw (-0.15,2) node[left]{$s$}--(0.15,2);
			\draw (7,-0.15) node[below]{$\ell$}--(7,0.15);

			\draw[-stealth](0,-1) -- (0,4) node[thick, scale=0.9,above] {speed};
			\draw[-stealth](0,0) -- (8.5,0) node[thick, scale=0.9,above] {length};

			 \draw (0.5,0.15)--(0.5,-0.15) node[below, scale=0.9]{$k_1$};
			 \draw (5.5,0.15)--(5.5,-0.15) node[below, scale=0.9]{$k_2$};

			\draw [draw=none, fill=gray, opacity=0.3] (s) -- (0.5,1.5) -- (5.5,1.5) -- (6,2) -- cycle;

			\draw [gray, line width=0.2mm, opacity=0.6] (0.5,1.5) -- (5.5,1.5);

			\draw [blue, line width=0.3mm] (s)--(dmin)--(6,2);
			\draw [darkspringgreen, line width=0.3mm] (s) -- (6,2);

			\draw [blue, dashed, line width=0.3mm, dash pattern=on 2pt off 2pt,dash phase=2pt] (6,2) -- (e);
			\draw [darkspringgreen, dashed, line width=0.3mm,dash pattern=on 2pt off 2pt] (6,2) -- (e); 

			\tikzstyle{every node}=[draw,circle,fill=white,inner sep=.8pt]

			\path (0,2) node{};
			\path (0.5,2) node{};
			\path (1,2) node{};
			\path (1.5,2) node{};
			\path (2,2) node{};
			\path (2.5,2) node{};
			\path (3,2) node{};
			\path (3.5,2) node{};
			\path (4,2) node{};
			\path (4.5,2) node{};
			\path (5,2) node{};
			\path (5.5,2) node{};
			\path (6,2) node{};
			
			\path (0.5,1.5) node{};			
			\path (1,1) node{};
			\path (1.5,0.5) node{};
			\path (2,0) node{};
			\path (2.5,-0.5) node{};
			\path (dmin) node{};
			\path (3.5,-0.5) node{};
			\path (4,0) node{};
			\path (4.5,0.5) node{};
			\path (5,1) node{};
			\path (5.5,1.5) node{};

			\path (6.5,2.5) node{};
			\path (e) node{};

		\end{tikzpicture}
	\end{subfigure}
	\begin{subfigure}[!t]{0.3\textwidth}
		\centering
		\begin{tikzpicture}[thick, scale=0.4, every node/.style={scale=0.6}]
			\coordinate (s) at (0,2);
			\coordinate (e) at (7,3);
			\coordinate (dmin) at (3,-1);

			\draw (-0.15,3) node[left]{$s'$}--(0.15,3);
			\draw (-0.15,2) node[left]{$s$}--(0.15,2);
			\draw (7,-0.15) node[below]{$\ell$}--(7,0.15);

			\draw[-stealth](0,-1) -- (0,4) node[thick, scale=0.9,above] {speed};
			\draw[-stealth](0,0) -- (8.5,0) node[thick, scale=0.9,above] {length};

			 \draw (1,0.15)--(1,-0.15) node[below, scale=0.9]{$k_1$};
			 \draw (5,0.15)--(5,-0.15) node[below, scale=0.9]{$k_2$};

			\draw [draw=none, fill=gray, opacity=0.3] (s) -- (1,1) -- (5,1) -- (6,2) -- cycle;

			\draw [gray, line width=0.2mm, opacity=0.6] (1,1) -- (5,1);

			\draw [blue, line width=0.3mm] (s)--(dmin)--(6,2);
			\draw [darkspringgreen, line width=0.3mm] (s) -- (6,2);

			\draw [blue, dashed, line width=0.3mm, dash pattern=on 2pt off 2pt,dash phase=2pt] (6,2) -- (e);
			\draw [darkspringgreen, dashed, line width=0.3mm,dash pattern=on 2pt off 2pt] (6,2) -- (e); 
			\tikzstyle{every node}=[draw,circle,fill=white,inner sep=.8pt]

			\path (0,2) node{};
			\path (0.5,2) node{};
			\path (1,2) node{};
			\path (1.5,2) node{};
			\path (2,2) node{};
			\path (2.5,2) node{};
			\path (3,2) node{};
			\path (3.5,2) node{};
			\path (4,2) node{};
			\path (4.5,2) node{};
			\path (5,2) node{};
			\path (5.5,2) node{};
			\path (6,2) node{};
			
			\path (0.5,1.5) node{};			
			\path (1,1) node{};
			\path (1.5,0.5) node{};
			\path (2,0) node{};
			\path (2.5,-0.5) node{};
			\path (dmin) node{};
			\path (3.5,-0.5) node{};
			\path (4,0) node{};
			\path (4.5,0.5) node{};
			\path (5,1) node{};
			\path (5.5,1.5) node{};

			\path (6.5,2.5) node{};
			\path (e) node{};
		\end{tikzpicture}
	\end{subfigure}
	\begin{subfigure}[!t]{0.3\textwidth}
		\centering
		\begin{tikzpicture}[thick, scale=0.4, every node/.style={scale=0.6}]
			\coordinate (s) at (0,2);
			\coordinate (e) at (7,3);
			\coordinate (dmin) at (3,-1);

			\draw (-0.15,3) node[left]{$s'$}--(0.15,3);
			\draw (-0.15,2) node[left]{$s$}--(0.15,2);
			\draw (7,-0.15) node[below]{$\ell$}--(7,0.15);

			\draw[-stealth](0,-1) -- (0,4) node[thick, scale=0.9,above] {speed};
			\draw[-stealth](0,0) -- (8.5,0) node[thick, scale=0.9,above] {length};

			 \draw (1.5,0.15)--(1.5,-0.15) node[below, scale=0.9]{$k_1$};
			 \draw (4.5,0.15)--(4.5,-0.15) node[below, scale=0.9]{$k_2$};

			\draw [draw=none, fill=gray, opacity=0.3] (s) -- (1.5,0.5) -- (4.5,0.5) -- (6,2) -- cycle;

			\draw [gray, line width=0.2mm, opacity=0.6] (1.5,0.5) -- (4.5,0.5);

			\draw [blue, line width=0.3mm] (s)--(dmin)--(6,2);
			\draw [darkspringgreen, line width=0.3mm] (s) -- (6,2);

			\draw [blue, dashed, line width=0.3mm, dash pattern=on 2pt off 2pt,dash phase=2pt] (6,2) -- (e);
			\draw [darkspringgreen, dashed, line width=0.3mm,dash pattern=on 2pt off 2pt] (6,2) -- (e); 
			\tikzstyle{every node}=[draw,circle,fill=white,inner sep=.8pt]

			\path (0,2) node{};
			\path (0.5,2) node{};
			\path (1,2) node{};
			\path (1.5,2) node{};
			\path (2,2) node{};
			\path (2.5,2) node{};
			\path (3,2) node{};
			\path (3.5,2) node{};
			\path (4,2) node{};
			\path (4.5,2) node{};
			\path (5,2) node{};
			\path (5.5,2) node{};
			\path (6,2) node{};
			
			\path (0.5,1.5) node{};			
			\path (1,1) node{};
			\path (1.5,0.5) node{};
			\path (2,0) node{};
			\path (2.5,-0.5) node{};
			\path (dmin) node{};
			\path (3.5,-0.5) node{};
			\path (4,0) node{};
			\path (4.5,0.5) node{};
			\path (5,1) node{};
			\path (5.5,1.5) node{};

			\path (6.5,2.5) node{};
			\path (e) node{};
		\end{tikzpicture}
	\end{subfigure}
	\caption{The lines represent the sequence of speeds realizing respectively distance $\drefmin$ (green) and $\dmin$ (blue). The removed portion area (gray) resulting from $\drefmin(\ell)$ 
	determine $k_1$ and $k_2$. Here, we remove respectively one, two, and three \emph{layers}.}
	\label{fig:removalarea}
\end{figure}

Observe that we can remove \textit{layers} of gray area (Figure \ref{fig:removalarea}), and this can be expressed using a compact formula: this is given by the following quadratic equation $-k^2+(\ell-|s'-s|)k$, whose derivation can be found in~\Cref{sec:appendix_existence}.

Let us assume that, by removing an integer number of $k$ layers, we match exactly $D$: if that's the case, the removed portion matches $D$, and the corresponding sequence of speeds realizing a distance $\delta$ 
coincides exactly with the orange trajectory, see Figure~\ref{fig:checkpoints} (left). 
However, one can notice that removing an integer number $k$ layers might not match exactly $D$ and exceed it: this happens indeed whenever removing $k-1$ layers is not enough (still not match $D$), while removing $k$ is too much (we 
exceed $D$); hence, we need to remove a slice of the area that we added with the $k$-th layer, but not all of it. In this case, the sequence of speeds realizing distance $\delta$ coincides with the orange 
trajectory, see Figure~\ref{fig:checkpoints} (right). In both cases, we can derive two values: $k_1=k$, consisting in the first \textit{turning} point (where by \textit{turning} we mean that the speed changes from constantly decelerating to stay put), and $k_2$, which will either be ``symmetric'' to $k_1$ (Figure~\ref{fig:checkpoints}, left case), or it can be obtained easily via residual calculation (Figure~\ref{fig:checkpoints}, right case) 

In both cases, the trajectory admits a compact representation using only a \emph{constant number} of intermediate configurations, namely:
\begin{equation}\label{eq:compact}
\begin{aligned}
	\hspace*{-.5cm}& \text{Figure~\ref{fig:checkpoints} (left): }  (-,k_1), (0,\ell-s'+s-2k_1), (+, s'-s+k_1) \\
	\hspace*{-.5cm}& \text{Figure~\ref{fig:checkpoints} (right): }  (-,k_1), (0,\ell-s'+s-2k_1-k_2), (+,1), (0, k_2), (+, s'-s+k_1-1) 
\end{aligned}
\end{equation}
where $(0,s)/(+,s)/(-,s)$ means keeping the same speed/accelerating/decelerating for $s$ successive configurations.
Such a compact representation can be computed in constant time for $\delta$ in $\mathcal{D}_\text{A}\cup\mathcal{D}_\text{B}\cup\mathcal{D}_\text{C}$ as shown in~\cref{lem:existence_k1k2} using similar ideas (see below).
\vspace{-5pt}
\begin{figure}[!htb]
	\centering
	\begin{subfigure}[!t]{0.49\textwidth}
		\centering
		\begin{tikzpicture}[thick, scale=0.5, every node/.style={scale=1}]
			\coordinate (s) at (0,2);
			\coordinate (e) at (7,3);
			\coordinate (dmin) at (3,-1);

			\draw (-0.15,3) node[left]{$s'$}--(0.15,3);
			\draw (-0.15,2) node[left]{$s$}--(0.15,2);
			\draw (7,-0.15) node[below]{$\ell$}--(7,0.15);

			\draw[-stealth](0,-1) -- (0,4) node[thick, scale=0.9,above] {speed};
			\draw[-stealth](0,0) -- (8.5,0) node[thick, scale=0.9,above] {length};

			\draw (1.5,0.15)--(1.5,-0.15) node[below, scale=0.8]{$k_1$};
			\draw (4.5,0.15)--(4.5,-0.15) node[below, scale=0.8]{$k_2$};

			\draw [draw=none, fill=gray, opacity=0.3] (s) -- (1.5,0.5) -- (4.5,0.5) -- (6,2) -- cycle;


			\draw [blue, line width=0.3mm] (1.5,0.5)--(dmin)--(4.5,0.5);
			\draw [darkspringgreen, line width=0.3mm] (s) -- (6,2);

			\draw [orange, dashed, line width=0.3mm, dash pattern=on 2pt off 2pt] (s) -- (1.5,0.5);
			\draw [blue, dashed, line width=0.3mm, dash pattern=on 2pt off 2pt,dash phase=2pt] (s) -- (1.5,0.5);

			\draw [orange, dashed, line width=0.3mm, dash pattern=on 2pt off 2pt] (4.5,0.5) -- (6,2);
			\draw [blue, dashed, line width=0.3mm, dash pattern=on 2pt off 2pt,dash phase=2pt] (4.5,0.5) -- (6,2);

			\draw [orange, line width=0.3mm] (1.5,0.5) -- (4.5,0.5);

			\draw [orange, dashed, line width=0.3mm, dash pattern=on 1pt off 2pt,dash phase=2pt] (6,2) -- (e);

			\draw [blue, dashed, line width=0.3mm, dash pattern=on 1pt off 2pt,dash phase=1pt] (6,2) -- (e);
			\draw [darkspringgreen, dashed, line width=0.3mm,dash pattern=on 1pt off 2pt] (6,2) -- (e);  

			\tikzstyle{every node}=[draw,circle,fill=white,inner sep=.8pt]

			\path (0,2) node{};
			\path (0.5,2) node{};
			\path (1,2) node{};
			\path (1.5,2) node{};
			\path (2,2) node{};
			\path (2.5,2) node{};
			\path (3,2) node{};
			\path (3.5,2) node{};
			\path (4,2) node{};
			\path (4.5,2) node{};
			\path (5,2) node{};
			\path (5.5,2) node{};
			\path (6,2) node{};
			
			\path (0.5,1.5) node{};			
			\path (1,1) node{};
			\path (1.5,0.5) node{};
			\path (2,0) node{};
			\path (2.5,-0.5) node{};
			\path (dmin) node{};
			\path (3.5,-0.5) node{};
			\path (4,0) node{};
			\path (4.5,0.5) node{};
			\path (5,1) node{};
			\path (5.5,1.5) node{};

			\path (6.5,2.5) node{};
			\path (e) node{};

		\end{tikzpicture}
	\end{subfigure}
	\begin{subfigure}[!t]{0.49\textwidth}
		\centering
		\begin{tikzpicture}[thick, scale=0.5, every node/.style={scale=1}]
			\coordinate (s) at (0,2);
			\coordinate (e) at (7,3);
			\coordinate (dmin) at (3,-1);

			\draw (-0.15,3) node[left]{$s'$}--(0.15,3);
			\draw (-0.15,2) node[left]{$s$}--(0.15,2);
			\draw (7,-0.15) node[below]{$\ell$}--(7,0.15);

			\draw[-stealth](0,-1) -- (0,4) node[thick, scale=0.9,above] {speed};
			\draw[-stealth](0,0) -- (8.5,0) node[thick, scale=0.9,above] {length};

			\draw (1.5,0.15)--(1.5,-0.15) node[below, scale=0.8]{$k_1$};
			\draw (3,0.15)--(3,-0.15) node[below, scale=0.8]{$k_2$};

			\draw [draw=none, fill=gray, opacity=0.3] (s) -- (1.5,0.5) -- (3,0.5) -- (3.5,1) -- (5,1) -- (6,2) -- cycle;


			\draw [blue, line width=0.3mm] (1.5,0.5)--(dmin)--(5,1);
			\draw [darkspringgreen, line width=0.3mm] (s) -- (6,2);

			\draw [orange, dashed, line width=0.3mm, dash pattern=on 2pt off 2pt] (s) -- (1.5,0.5);
			\draw [blue, dashed, line width=0.3mm, dash pattern=on 2pt off 2pt,dash phase=2pt] (s) -- (1.5,0.5);

			\draw [orange, dashed, line width=0.3mm, dash pattern=on 2pt off 2pt] (5,1) -- (6,2);
			\draw [blue, dashed, line width=0.3mm, dash pattern=on 2pt off 2pt,dash phase=2pt] (5,1) -- (6,2);

			\draw [orange, line width=0.3mm] (1.5,0.5) --(3,0.5) -- (3.5,1) -- (5,1);

			\draw [orange, dashed, line width=0.3mm, dash pattern=on 1pt off 2pt,dash phase=2pt] (6,2) -- (e);

			\draw [blue, dashed, line width=0.3mm, dash pattern=on 1pt off 2pt,dash phase=1pt] (6,2) -- (e);
			\draw [darkspringgreen, dashed, line width=0.3mm,dash pattern=on 1pt off 2pt] (6,2) -- (e);  

			\tikzstyle{every node}=[draw,circle,fill=white,inner sep=.8pt]

			\path (0,2) node{};
			\path (0.5,2) node{};
			\path (1,2) node{};
			\path (1.5,2) node{};
			\path (2,2) node{};
			\path (2.5,2) node{};
			\path (3,2) node{};
			\path (3.5,2) node{};
			\path (4,2) node{};
			\path (4.5,2) node{};
			\path (5,2) node{};
			\path (5.5,2) node{};
			\path (6,2) node{};
			
			\path (0.5,1.5) node{};			
			\path (1,1) node{};
			\path (1.5,0.5) node{};
			\path (2,0) node{};
			\path (2.5,-0.5) node{};
			\path (dmin) node{};
			\path (3.5,-0.5) node{};
			\path (4,0) node{};
			\path (4.5,0.5) node{};
			\path (5,1) node{};
			\path (5.5,1.5) node{};

			\path (6.5,2.5) node{};
			\path (e) node{};
		\end{tikzpicture}
	\end{subfigure}
	\caption{The gray area removed from $\drefmin(\ell)$ that realizes the desired distance $\delta$ (orange).}
	\label{fig:checkpoints}
\end{figure}
\vspace{-5pt}

\begin{theorem}\label{lem:existence_k1k2}
Let $c=(0,s)$ and $c'=(\delta,s')$, where $s,s',\delta\in\mathbb{Z}_{\geq0}$. For every integer $\ell\geq |s'-s|$ satisfying $\dmin(\ell)\leq\delta\leq\dmax(\ell)$, a trajectory of length $\ell$ from $c$ to $c'$ can be constructed in constant time and represented by at most five monotonic segments.
\end{theorem}

\newpage

\bibliographystyle{splncs04}
\bibliography{paper}

\newpage

\appendix

\section{Missing content of~\Cref{sec:two-point}}

\subsection{Proof of Theorem~\ref{lem:formula}}
\label{sec:appendix_computation}

The goal of this section is to prove the following lemma from Section~\ref{sec:main-case}:

\Formula*

The proof follows two steps. First we compute expressions for $\dmin(t)$ and $\dmax(t)$ in the various cases where $s\le s'$ or $s>s'$, and where $t$ is even or odd. The second step is routine calculation to show that these expressions reduce to the ones given in the statement. It is a general comment that to pass from speed $s$ to $s'$ requires at least $\mid s- s'\mid$ time steps, called $t_{min}$ in~Section~\ref{sec:main-case}. Thus, our expressions of $\dmin(t)$ and $\dmax(t)$ will be interpreted only in the case that $t\ge t_{min}$.

\subsubsection{Details of the calculation of $\dmin(t)$}
\noindent{\bf Case 1 ($s\le s'$)}: In general, for $t\ge s'-s$ the trajectory leading to $\dmin(t)$ is made of two parts. First, $t-s'+s$ steps are used to decrease the speed to a minimal value and then increasing the speed to get back to speed $s$. In the second part, $s'-s$ time steps are used to pass from speed $s$ to $s'$.

If $t-s'+s$ is even the speed goes from $s$ to $s-\frac{t-s'+s}{2}=-\frac{t-s'-s}{2}$ and back to $s$. The trajectory spans a total distance of 
\begin{equation}\label{evencase}
\frac{t-s'+s}{2}\biggl(\frac{3s+s'-t}{2}\biggr).
\end{equation}
 If $t-s'+s$ is odd the speed goes from $s$ to $s-\lfloor\frac{t-s'+s}{2}\rfloor=s-\frac{t-1-s'+s}{2}$ and  back to $s$. The trajectory spans a distance of
\begin{equation}\label{oddcase} s+\frac{t-1-s'+s}{2}\biggl(\frac{3s-t+s'-1}{2}\biggr).
\end{equation}

Equation $(\ref{evencase})$ comes from the summation
\begin{equation*}
\underbrace{(s-1)+\ldots+(s-\frac{t-s'+s}{2})}_{\frac{t-s'+s}{2}\text{ terms}}+\underbrace{(s-\frac{t-s'+s}{2}+1)+\ldots+s}_{\frac{t-s'+s}{2}\text{ terms}},
\end{equation*}
and Equation $(\ref{oddcase})$  from
\begin{equation*}
\underbrace{(s-1)+\ldots+(s-\frac{t-1-s'+s}{2})}_{\frac{t-1-s'+s}{2}\text{ terms}}+(s-\frac{t-1-s'+s}{2})+\underbrace{(s-\frac{t-1-s'+s}{2}+1)+\ldots+s}_{\frac{t-1-s'+s}{2}\text{ terms}}.
\end{equation*}

The last $s'-s$ time units are used to pass from speed $s$ to $s'$ and impact a deplacement of $\dmin(s'-s)= (s+1) + \ldots + (s+(s'-s))=(s'-s)(s'+s+1)/2$. Adding this to Equations $(\ref{evencase})$ and $(\ref{oddcase})$ leads to $\dmin(t)$.

$$\dmin(t)=\begin{cases}0 & t<s'-s\\ \dmin(s'-s)+\frac{t-s'+s}{2}\biggl(\frac{3s+s'-t}{2}\biggr) & t-s'+s\text{ is even} \\ \dmin(s'-s)+s+\frac{t-1-s'+s}{2}\biggl(\frac{3s-t+s'-1}{2}\biggr)& t-s'+s\text{ is odd}   \end{cases}$$

We notice that $\dmin(s'-s)=\alpha$ with $\alpha$ given in the statement of the lemma. Moreover, direct calculations lead to
$$ \frac{t-s'+s}{2}\biggl(\frac{3s+s'-t}{2}\biggr)= s+\frac{t-1-s'+s}{2}\biggl(\frac{3s-t+s'-1}{2}\biggr) - \frac{1}{4},$$
which shows that the expression for $\dmin(t)$ given for $t-s'+s$ is even can be used to compute $\dmin(t)$ when $t-s'+s$ is odd provided we ceil the result. This shows that the lemma holds when $s'\ge s$. \medskip

\noindent{\bf Case 2 ($s\ge s'$)}: In this case, the minimum number of time units to pass from speed $s$ to $s'$ is $s-s'$ and this needs a distance of 
$$\dmin(s-s')=(s-1)+\ldots+s'=(s-s')(s+s'-1)/2.$$
We proceed wimilarly than for the case 1 and get the formula
$$\dmin(t)=\begin{cases}0 & t<s-s'\\ \dmin(s-s')+\frac{t-s+s'}{2}\biggl(\frac{3s'+s-t}{2}\biggr) & t-s+s'\text{ is even} \\ \dmin(s-s')'+s'+\frac{t-1-s+s'}{2}\biggl(\frac{3s'-t+s-1}{2}\biggr)& t-s+s'\text{ is odd}   \end{cases}$$
which is the same (although in a different form) as the one we get for $s'\ge s$ and the statement of the lemma is then vefified in the case $s\ge s'$.

\subsubsection{Details of the calculation of $\dmax(t)$}
\noindent{\bf Case 1 ($s\le s'$)}: To pass from speed $s$ to $s'$ needs $s'-s$ time units and distance $\dmax(s'-s)=(s+1)+\ldots+s'=(s'-s)(s+s'+1)/2$. After reaching speed $s'$, the car accelerate to a maximal speed and back to speed $s'$ if $t-s+s'>0$. 

If $t-s'+s$ is even the distance traveled from speed $s'$ to a maximal speed and then back to speed $s$ is
\begin{equation*}
\underbrace{(s'+1)+\ldots+(s'+\frac{t-s'+s}{2})}_{\frac{t-s'+s}{2}\text{ terms}}+\underbrace{(s'+\frac{t-s'+s}{2}-1)+\ldots+s'}_{\frac{t-s'+s}{2}\text{ terms}},
\end{equation*}
which evaluates to
\begin{equation}\label{maxeven}
\frac{t-s'+s}{2}\biggl(\frac{t+3s'+s}{2}\biggr).
\end{equation} 
If $t-s'+s$ is odd the distance traveled from speed $s'$ to a maximal speed and then back to speed $s$ is
\begin{equation*}
\underbrace{(s'+1)+\ldots+(s'+\frac{t-1-s'+s}{2})}_{\frac{t-1-s'+s}{2}\text{ terms}}+(s'+\frac{t-1-s'+s}{2})+\underbrace{(s'+\frac{t-1-s'+s}{2}-1)+\ldots+s'}_{\frac{t-1-s'+s}{2}\text{ terms}},
\end{equation*}
which leads to the following closed form
\begin{equation}\label{maxodd}
\frac{t+1-s'+s}{2}\biggl(\frac{t-1+3s'+s}{2}\biggr)-s'
\end{equation}
To summarize:
$$\dmax(t)=\begin{cases}0 & t<s'-s\\ \dmax(s'-s)+\frac{t-s'+s}{2}\biggl(\frac{t+3s'+s}{2}\biggr) & t-s+s'\text{ is even} \\ \dmax(s'-s)+\frac{t+1-s'+s}{2}\biggl(\frac{t-1+3s'+s}{2}\biggr)-s'& t-s+s'\text{ is odd}   \end{cases}$$
We see that the expression for $\dmax(t)$ when $t-s+s'$ is even is the one given in the lemma statement. To handle the case when $t-s+s'$ is odd, we compute
$$\frac{t-s'+s}{2}\biggl(\frac{t+3s'+s}{2}\biggr)-\frac{t+1-s'+s}{2}\biggl(\frac{t-1+3s'+s}{2}\biggr)+s'=\frac{1}{4}$$
which shows that if the expression for $t-s+s'$ is used to compute $\dmax(t)$ when $t-s+s'$ is odd and then flooring leads to the right result as stated in the statment of the lemma.
\medskip

\noindent{\bf Case 2 ($s\ge s'$)}: In this case we proceed similarly and we do not reproduce the detailed calculation here. The final result is
$$\dmax(t)=\begin{cases}0 &  t<s-s'\\ \dmax(s-s')+\frac{t-s+s'}{2}\biggl(\frac{t+3s+s'}{2}\biggr) & t-s+s'\text{ is even} \\ \dmax(s-s')+\frac{t+1-s+s'}{2}\biggl(\frac{t-1+3s+s'}{2}\biggr)-s& t-s+s'\text{ is odd}   \end{cases}$$
with $\dmax(s-s')=(s-s')(s+s'-1)/2$, which leads to the statement of the lemma.

\section{Missing content of~\Cref{sec:BT}}
\subsection{Proof of Theorem~\ref{lem:existence_k1k2}}
\label{sec:appendix_existence}

Let us fix $s,s',\delta\in\mathbb{Z}_{\geq0}$ and an integer $\ell\geq |s'-s|$ satisfying $\dmin(\ell)\leq\delta\leq\dmax(\ell)$. The three ranges $\mathcal{D}_\text{A}$, $\mathcal{D}_\text{B}$, and $\mathcal{D}_\text{C}$ cover this entire interval (at the boundaries, one can use either case).

Let us start by explicating the computation of the gray area. Fix an integer $k\geq0$ and denote with $A(k)$ the amount of area in {\bf A} where we remove \emph{layers} from $\drefmin$

\begin{figure}[!htb]
	\centering
	\begin{tikzpicture}[thick, scale=0.7, every node/.style={scale=1}]
		\coordinate (s) at (0,2);
		\coordinate (e) at (7,3);
		\coordinate (dmin) at (3,-1);

		\draw (-0.15,3) node[left]{$s'$}--(0.15,3);
		\draw (-0.15,2) node[left]{$s$}--(0.15,2);
		\draw (7,-0.15) node[below]{$\ell$}--(7,0.15);

		\draw[-stealth](0,-1) -- (0,4) node[thick, scale=0.9,above] {speed};
		\draw[-stealth](0,0) -- (8.5,0) node[thick, scale=0.9,above] {length};

		\draw (1.5,0.15)--(1.5,-0.15) node[below, scale=0.6]{$k$};

		\draw [draw=none, fill=gray, opacity=0.3] (s) -- (1.5,0.5) -- (4.5,0.5) -- (6,2) -- cycle;

		\node [color=gray] at (3,1) {$A(k)$};

		\draw [blue, line width=0.3mm] (1.5,0.5)--(dmin)--(4.5,0.5);
		\draw [darkspringgreen, line width=0.3mm] (s) -- (6,2);

		\draw [orange, dashed, line width=0.3mm, dash pattern=on 2pt off 2pt] (s) -- (1.5,0.5);
		\draw [blue, dashed, line width=0.3mm, dash pattern=on 2pt off 2pt,dash phase=2pt] (s) -- (1.5,0.5);

		\draw [orange, dashed, line width=0.3mm, dash pattern=on 2pt off 2pt] (4.5,0.5) -- (6,2);
		\draw [blue, dashed, line width=0.3mm, dash pattern=on 2pt off 2pt,dash phase=2pt] (4.5,0.5) -- (6,2);

		\draw [orange, line width=0.3mm] (1.5,0.5) -- (4.5,0.5);

		\draw [orange, dashed, line width=0.3mm, dash pattern=on 1pt off 2pt,dash phase=2pt] (6,2) -- (e);

		\draw [blue, dashed, line width=0.3mm, dash pattern=on 1pt off 2pt,dash phase=1pt] (6,2) -- (e);
		\draw [darkspringgreen, dashed, line width=0.3mm,dash pattern=on 1pt off 2pt] (6,2) -- (e);  

	\end{tikzpicture}
	\caption{The symmetric amount of area $A(k)$ that we can remove from $\drefmin$}
	\label{fig:akfigure}
\end{figure}

First of all we can express the amount $A(k)$ (Figure \ref{fig:akfigure}).
\begin{align*}
A(k) &= 1 + 2 + \ldots + (k) + \left( \ell-|s'-s| - 2k \right)k + (k-1) + \ldots + 2 + 1 \\  
 &= \frac{k(k+1)}{2} + \left( \ell-|s'-s| - 2k \right)k + \frac{k(k-1)}{2} \\
 &= -k^2 + k(\ell-|s'-s|)
\end{align*}
$A(k)$ satisfies different properties:
\begin{itemize}
\item $A(k)$ is a concave quadratic function;
\item $A(k)$ attains its maximum at $k_m=(\ell-|s'-s|)/2 \geq 0$ with value $A(k_m) = \frac{(\ell-|s'-s|)^2}{4}$;
\item $A(0) = 0$ and $A(k)$ is increasing on $[0, k_m]$;
\end{itemize}

\begin{remark}
For the case $\delta \in {\bf B}$, if $\ell = |s'-s|$, we return the sequence $(+,\ell)$ when
$s'>s$, or $(-,\ell)$ when $s'<s$. If $\ell=0$, we return
the empty sequence of moves. Henceforth, we thus assume
$\ell > |s'-s|$, so the denominators used in Case~2 (below)
are strictly positive.
\end{remark}

\begin{remark}
We slightly abuse notation: $k_1$ will be used both to denotes the corresponding real root and its rounded integer value, 
as specified in each case. Similarly, $k_2$ denotes either a breakpoint time in the geometric description or a residual duration 
in the formulas. The explicit compact descriptions clarify the intended meaning.
\end{remark}

\noindent{\bf Case 1 ($\delta \in$ {\bf A})}: we can express the \emph{true} integer distance $D = \drefmin(\ell) - \delta \ge 0$. Since $\delta \in $ {\bf A}, we have
\[ 0\le D\le \drefmin(\ell)-\dmin(\ell) =\left\lfloor \frac{(\ell-|s'-s|)^2}{4} \right\rfloor =A(\lfloor (\ell-|s'-s|)/2\rfloor) \le A(k_m). \]
Thus $A(k)=D$ has a smaller root $\in[0,\lfloor (\ell-|s'-s|)/2\rfloor]$, computable by solving the corresponding quadratic equation.

\begin{figure}[!htb]
	\centering
	\begin{subfigure}[!t]{0.49\textwidth}
		\centering
		\begin{tikzpicture}[thick, scale=0.6, every node/.style={scale=1}]
			\coordinate (s) at (0,2);
			\coordinate (e) at (7,3);
			\coordinate (dmin) at (3,-1);

			\draw (-0.15,3) node[left]{$s'$}--(0.15,3);
			\draw (-0.15,2) node[left]{$s$}--(0.15,2);
			\draw (7,-0.15) node[below]{$\ell$}--(7,0.15);

			\draw[-stealth](0,-1) -- (0,4) node[thick, scale=0.9,above] {speed};
			\draw[-stealth](0,0) -- (8.5,0) node[thick, scale=0.9,above] {length};

			\draw (1.5,0.15)--(1.5,-0.15) node[below, scale=0.6]{$k_1$};
			\draw (4.5,0.15)--(4.5,-0.15) node[below, scale=0.6]{$k_2$};

			\draw [draw=none, fill=gray, opacity=0.3] (s) -- (1.5,0.5) -- (4.5,0.5) -- (6,2) -- cycle;


			\draw [blue, line width=0.3mm] (1.5,0.5)--(dmin)--(4.5,0.5);
			\draw [darkspringgreen, line width=0.3mm] (s) -- (6,2);

			\draw [orange, dashed, line width=0.3mm, dash pattern=on 2pt off 2pt] (s) -- (1.5,0.5);
			\draw [blue, dashed, line width=0.3mm, dash pattern=on 2pt off 2pt,dash phase=2pt] (s) -- (1.5,0.5);

			\draw [orange, dashed, line width=0.3mm, dash pattern=on 2pt off 2pt] (4.5,0.5) -- (6,2);
			\draw [blue, dashed, line width=0.3mm, dash pattern=on 2pt off 2pt,dash phase=2pt] (4.5,0.5) -- (6,2);

			\draw [orange, line width=0.3mm] (1.5,0.5) -- (4.5,0.5);

			\draw [orange, dashed, line width=0.3mm, dash pattern=on 1pt off 2pt,dash phase=2pt] (6,2) -- (e);

			\draw [blue, dashed, line width=0.3mm, dash pattern=on 1pt off 2pt,dash phase=1pt] (6,2) -- (e);
			\draw [darkspringgreen, dashed, line width=0.3mm,dash pattern=on 1pt off 2pt] (6,2) -- (e);  
		\end{tikzpicture}
	\end{subfigure}
	\begin{subfigure}[!t]{0.49\textwidth}
		\centering
		\begin{tikzpicture}[thick, scale=0.6, every node/.style={scale=1}]
			\coordinate (s) at (0,2);
			\coordinate (e) at (7,3);
			\coordinate (dmin) at (3,-1);

			\draw (-0.15,3) node[left]{$s'$}--(0.15,3);
			\draw (-0.15,2) node[left]{$s$}--(0.15,2);
			\draw (7,-0.15) node[below]{$\ell$}--(7,0.15);

			\draw[-stealth](0,-1) -- (0,4) node[thick, scale=0.9,above] {speed};
			\draw[-stealth](0,0) -- (8.5,0) node[thick, scale=0.9,above] {length};

			\draw (1.5,0.15)--(1.5,-0.15) node[below, scale=0.6]{$k_1$};
			\draw (3,0.15)--(3,-0.15) node[below, scale=0.6]{$k_2$};

			\draw [draw=none, fill=gray, opacity=0.3] (s) -- (1.5,0.5) -- (3,0.5) -- (3.5,1) -- (5,1) -- (6,2) -- cycle;


			\draw [blue, line width=0.3mm] (1.5,0.5)--(dmin)--(5,1);
			\draw [darkspringgreen, line width=0.3mm] (s) -- (6,2);

			\draw [orange, dashed, line width=0.3mm, dash pattern=on 2pt off 2pt] (s) -- (1.5,0.5);
			\draw [blue, dashed, line width=0.3mm, dash pattern=on 2pt off 2pt,dash phase=2pt] (s) -- (1.5,0.5);

			\draw [orange, dashed, line width=0.3mm, dash pattern=on 2pt off 2pt] (5,1) -- (6,2);
			\draw [blue, dashed, line width=0.3mm, dash pattern=on 2pt off 2pt,dash phase=2pt] (5,1) -- (6,2);

			\draw [orange, line width=0.3mm] (1.5,0.5) --(3,0.5) -- (3.5,1) -- (5,1);

			\draw [orange, dashed, line width=0.3mm, dash pattern=on 1pt off 2pt,dash phase=2pt] (6,2) -- (e);

			\draw [blue, dashed, line width=0.3mm, dash pattern=on 1pt off 2pt,dash phase=1pt] (6,2) -- (e);
			\draw [darkspringgreen, dashed, line width=0.3mm,dash pattern=on 1pt off 2pt] (6,2) -- (e);  

		\end{tikzpicture}
	\end{subfigure}
	\caption{Sequences of speeds realizing respectively distance $\drefmin$ (green), $\delta$ (orange) and $\dmin$ (blue). The gray portion is the amount of area removed from $\drefmin$ in order to match 
	$D = \drefmin(\ell)-\delta$. Such construction always yields two intermediate configurations $k_1$ and $k_2$.}
	\label{fig:checkpointsA}
\end{figure}

We proceed by solving the equation $A(k)=D$.
By taking the smaller root of this equation $k_1$, we distinguish between two cases:
\begin{itemize}
\item $k_1 \in \mathbb{Z}_{\ge 0}$: if $k_1$ is integer, then $A(k)$ matches exactly $D$; therefore, the equation $A(k)$ above coincides exactly with the amount of area to remove from $\drefmin$ in order to 
match $\delta$ (left scenario of Figure \ref{fig:checkpointsA}).\\
The second checkpoint is then given by the unique integer value that {\itshape aligns} with the $\dmin$ trajectory. This can be straightforwardly computed as $k_2 = \ell - |s' - s| - k_1$;
\item $k_1 \in \mathbb{R}_{> 0}$: if $k_1$ is real, then $A(k)$ does not match exactly $D$, and we take $\lceil k_1 \rceil$. This implies that $A(\lceil k_1 \rceil) - D > 0$. Furthermore, by construction, $A(\lceil k_1 \rceil-1) < D < A(\lceil k_1 \rceil)$,
hence there is an excess area $A(\lceil k_1\rceil)-D$. Therefore, there exists another integer value in the range $\left[ \lceil k_1 \rceil, \ell -|s' - s|- \lceil k1 \rceil \right]$ where we need 
to \emph{decrease} the area of $A(k)$ in order to match exactly $D$ (right scenario of Figure \ref{fig:checkpointsA}). This adjustment is described by the residual $k_2 = A(\lceil k_1 \rceil)-D$, which is a valid integer value since both 
$A(\lceil k_1 \rceil)$ and $D$ are integers and lies in the range $(0,\ell -|s' - s|- 2\lceil k1 \rceil]$. 
\end{itemize}

The above holds independently of whether $s\leq s'$ or $s > s'$. We need to distinguish such case whenever we describe the trajecory itself: indeed, such a trajectory can be described in a compact and deterministic way by a sequence of 
pairs $(a,b)$, where $a \in \{-,0,+\}$ (which stands for decelerate, stay at the same speed, accelerate) and $b\in \mathbb{Z}_{\geq 0}$, telling us for how many steps we should repeat action $a$.
\begin{itemize} 
	\item When $s\leq s'$:
	\begin{enumerate}
		\item In the first case ($k_1$ integer), we have\\
		$(-,k_1), (0,\ell-s'+s-2k_1), (+, s'-s+k_1)$
		\item In the second case, we have\\
		$(-,\lceil k_1 \rceil), (0,\ell-s'+s-2\lceil k_1 \rceil-k_2), (+,1), (0, k_2), (+, s'-s+\lceil k_1 \rceil-1)$
	\end{enumerate}
	\item When $s>s'$:
	\begin{enumerate}
		\item In the first case ($k_1$ integer), we have\\
		$(-,s-s'+k_1), (0,\ell-s+s'-2k_1), (+, k_1)$
		\item In the second case, we have\\
		$(-,s-s'+\lceil k_1 \rceil), (0,\ell-s+s'-2\lceil k_1 \rceil-k_2), (+,1), (0,k_2), (+, \lceil k_1 \rceil-1)$
	\end{enumerate}
\end{itemize} 

\noindent{\bf Case 2 ($\delta \in$ {\bf B} and $s \leq s'$)}:

\begin{figure}[!htb]
	\centering
	\begin{subfigure}[!t]{0.49\textwidth}
		\centering
		\begin{tikzpicture}[thick, scale=0.6, every node/.style={scale=1}]
			\coordinate (s) at (0,1);
			\coordinate (e) at (7,4);

			\draw (-0.15,4) node[left]{$s'$}--(0.15,4);
			\draw (-0.15,1) node[left]{$s$}--(0.15,1);
			\draw (7,-0.15) node[below]{$\ell$}--(7,0.15);

			\draw[-stealth](0,-1) -- (0,5) node[thick, scale=0.9,above] {speed};
			\draw[-stealth](0,0) -- (8.5,0) node[thick, scale=0.9,above] {length};

			\draw (1.5,0.15)--(1.5,-0.15) node[below, scale=0.6]{$k_1$};
			\draw (5.5,0.15)--(5.5,-0.15) node[below, scale=0.6]{$k_2$};

			\draw [draw=none, fill=gray, opacity=0.3] (s) -- (1.5,2.5) -- (5.5,2.5) -- (4,1) -- cycle;


			\draw [color=brown, line width=0.3mm] (1.5,2.5) -- (3,4) -- (e);
			\draw [color=darkspringgreen, line width=0.3mm] (s) -- (4,1) -- (5.5,2.5);
			\draw [color=orange, line width=0.3mm] (1.5,2.5) -- (5.5,2.5);
			
			\draw [color=brown, dashed, line width=0.3mm, dash pattern=on 2pt off 2pt,dash phase=2pt] (s) -- (1.5,2.5);
			\draw [color=orange, dashed,line width=0.3mm, dash pattern=on 2pt off 2pt] (s) -- (1.5,2.5);
			
			\draw [color=darkspringgreen, dashed, line width=0.3mm, dash pattern=on 2pt off 2pt,dash phase=2pt] (5.5,2.5) -- (e);
			\draw [color=orange, dashed,line width=0.3mm, dash pattern=on 2pt off 2pt] (5.5,2.5) -- (e);

		\end{tikzpicture}
	\end{subfigure}
	\begin{subfigure}[!t]{0.49\textwidth}
		\centering
		\begin{tikzpicture}[thick, scale=0.6, every node/.style={scale=1}]
			\coordinate (s) at (0,1);
			\coordinate (e) at (7,4);

			\draw (-0.15,4) node[left]{$s'$}--(0.15,4);
			\draw (-0.15,1) node[left]{$s$}--(0.15,1);
			\draw (7,-0.15) node[below]{$\ell$}--(7,0.15);

			\draw[-stealth](0,-1) -- (0,5) node[thick, scale=0.9,above] {speed};
			\draw[-stealth](0,0) -- (8.5,0) node[thick, scale=0.9,above] {length};

			\draw (1.5,0.15)--(1.5,-0.15) node[below, scale=0.6]{$k_1$};
			\draw (3.5,0.15)--(3.5,-0.15) node[below, scale=0.6]{$k_2$};

			\draw [draw=none, fill=gray, opacity=0.3] (s) -- (1.5,2.5) -- (3.5,2.5) -- (4,3) -- (6,3) -- (4,1) -- cycle;


			\draw [color=brown, line width=0.3mm] (1.5,2.5) -- (3,4) -- (e);
			\draw [color=darkspringgreen, line width=0.3mm] (s) -- (4,1) -- (6,3);
			\draw [color=orange, line width=0.3mm] (1.5,2.5) -- (3.5,2.5) -- (4,3) -- (6,3);
			
			\draw [color=brown, dashed, line width=0.3mm, dash pattern=on 2pt off 2pt,dash phase=2pt] (s) -- (1.5,2.5);
			\draw [color=orange, dashed,line width=0.3mm, dash pattern=on 2pt off 2pt] (s) -- (1.5,2.5);
			
			\draw [color=darkspringgreen, dashed, line width=0.3mm, dash pattern=on 2pt off 2pt,dash phase=2pt] (6,3) -- (e);
			\draw [color=orange, dashed,line width=0.3mm, dash pattern=on 2pt off 2pt] (6,3) -- (e);
		\end{tikzpicture}
	\end{subfigure}
	\caption{Two possible trajectories of distance $\delta$ in the {\bf B} region in the case $s \leq s'$. The left scenario happens whenever $D$ matches exactly $A(k)$, while the right scenario whenever $A(k)$ is an \emph{underestimation} of $D$.}
	\label{fig:checkpointsBone}
\end{figure}

We can express the \emph{true} integer distance as $D = \delta-\drefmin(\ell)$. In this case, we can express the amount $A(k)$ linearly, just by counting the internal amount done at 
$k$ for $\ell-s'+s$ times, i.e. $A(k) = (\ell-s'+s)k$. We can solve the equation $A(k)=D$, take the only root of this equation $k_1 = D/(\ell-s'+s)$ and distinguish between two cases:
\begin{itemize}
\item $k_1 \in \mathbb{Z}_{\geq0}$, then $A(k)$ matches exactly $D$ (left Figure \ref{fig:checkpointsBone}); therefore we accelerate for $k_1$ steps, and then keep the same velocity. The second checkpoint is then given by the unique integer value that {\itshape aligns} with the $\drefmin$ trajectory. This can be straightforwardly computed by $k_2 = \ell - s' + s + k_1$;
\item $k_1 \in \mathbb{R}_{>0}$, then we take $\lfloor k_1 \rfloor$ (right Figure \ref{fig:checkpointsBone}); similarly we have $A(\lfloor k_1 \rfloor) < D < A(\lfloor k_1 \rfloor + 1)$; this time, we need to add an extra area portion by doing a 1 step acceleration in a certain integer point between $\lfloor k_1 \rfloor$ and $\ell-s'+s+\lfloor k_1 \rfloor$. This is given by  $k_2 = D - (\lfloor k_1 \rfloor)(\ell-s'+s)$. 
\end{itemize}

Hence we can describe such a trajectory in a compact and deterministic way by a sequence of pairs $(a,b)$:
\begin{enumerate}
	\item In the first case, we have $(+,k_1), (0,\ell-s'+s), (+, s'-s-k_1)$
	\item Otherwise, we have $(+,\lfloor k_1 \rfloor), (0,\ell-s'+s-k_2), (+,1), (0, k_2), (+, s'-s-\lfloor k_1 \rfloor-1)$
\end{enumerate}

\noindent{\bf Case 2 ($\delta \in$ {\bf B} and $s > s'$)}: 

\begin{figure}[!htb]
	\centering
	\begin{subfigure}[!t]{0.49\textwidth}
		\centering
		\begin{tikzpicture}[thick, scale=0.6, every node/.style={scale=1}]
			\coordinate (s) at (0,4);
			\coordinate (e) at (7,1);

			\draw (-0.15,1) node[left]{$s'$}--(0.15,1);
			\draw (-0.15,4) node[left]{$s$}--(0.15,4);
			\draw (7,-0.15) node[below]{$\ell$}--(7,0.15);

			\draw[-stealth](0,-1) -- (0,5) node[thick, scale=0.9,above] {speed};
			\draw[-stealth](0,0) -- (8.5,0) node[thick, scale=0.9,above] {length};

			\draw (1,0.15)--(1,-0.15) node[below, scale=0.6]{$k_1$};
			\draw (5,0.15)--(5,-0.15) node[below, scale=0.6]{$k_2$};

			\draw [draw=none, fill=gray, opacity=0.3] (3,1) -- (1,3) -- (5,3) -- (e) -- cycle;


			\draw [color=brown, line width=0.3mm] (s) -- (4,4) -- (5,3);
			\draw [color=darkspringgreen, line width=0.3mm] (1,3) -- (3,1) -- (e);
			\draw [color=orange, line width=0.3mm] (1,3) -- (5,3);
			
			\draw [color=brown, dashed, line width=0.3mm, dash pattern=on 2pt off 2pt,dash phase=2pt] (5,3) -- (e);
			\draw [color=orange, dashed,line width=0.3mm, dash pattern=on 2pt off 2pt] (5,3) -- (e);
			
			\draw [color=darkspringgreen, dashed, line width=0.3mm, dash pattern=on 2pt off 2pt,dash phase=2pt] (s) -- (1,3);
			\draw [color=orange, dashed,line width=0.3mm, dash pattern=on 2pt off 2pt] (s) -- (1,3);

		\end{tikzpicture}
	\end{subfigure}
	\begin{subfigure}[!t]{0.49\textwidth}
		\centering
		\begin{tikzpicture}[thick, scale=0.6, every node/.style={scale=1}]
			\coordinate (s) at (0,4);
			\coordinate (e) at (7,1);

			\draw (-0.15,1) node[left]{$s'$}--(0.15,1);
			\draw (-0.15,4) node[left]{$s$}--(0.15,4);
			\draw (7,-0.15) node[below]{$\ell$}--(7,0.15);

			\draw[-stealth](0,-1) -- (0,5) node[thick, scale=0.9,above] {speed};
			\draw[-stealth](0,0) -- (8.5,0) node[thick, scale=0.9,above] {length};

			\draw (1,0.15)--(1,-0.15) node[below, scale=0.6]{$k_1$};
			\draw (3,0.15)--(3,-0.15) node[below, scale=0.6]{$k_2$};

			\draw [draw=none, fill=gray, opacity=0.3] (3,1) -- (1,3) -- (3,3) -- (3.5,2.5) -- (5.5,2.5) -- (e) -- cycle;


			\draw [color=brown, line width=0.3mm] (s) -- (4,4) -- (5.5,2.5);
			\draw [color=darkspringgreen, line width=0.3mm] (1,3) -- (3,1) -- (e);
			\draw [color=orange, line width=0.3mm] (1,3) -- (3,3) -- (3.5,2.5) -- (5.5,2.5);
			
			\draw [color=brown, dashed, line width=0.3mm, dash pattern=on 2pt off 2pt,dash phase=2pt] (5.5,2.5) -- (e);
			\draw [color=orange, dashed,line width=0.3mm, dash pattern=on 2pt off 2pt] (5.5,2.5) -- (e);
			
			\draw [color=darkspringgreen, dashed, line width=0.3mm, dash pattern=on 2pt off 2pt,dash phase=2pt] (s) -- (1,3);
			\draw [color=orange, dashed,line width=0.3mm, dash pattern=on 2pt off 2pt] (s) -- (1,3);

		\end{tikzpicture}
	\end{subfigure}
	\caption{Two possible trajectories of distance $\delta$ in the {\bf B} region in the case $s > s'$. The left scenario happens whenever $D$ matches exactly $A(k)$, while the right scenario whenever $A(k)$ is an \emph{overestimation} of $D$}
	\label{fig:checkpointsBtwo}
\end{figure}

We can express the \emph{true} integer distance as $D = \drefmax(\ell)-\delta$. Similarly as before, we can express the excess $A(k)$ just by counting the internal amount done at $k$ for $\ell-s+s'$ times, i.e. $A(k) = (\ell-s+s')k$. 
We can solve the equation $A(k)=D$, take the only root of this equation $k_1 = D/(\ell-s+s')$ and distinguish between two cases:
\begin{itemize}
\item $k_1 \in \mathbb{Z}_{\geq0}$, then $A(k)$ matches exactly $D$ (left Figure \ref{fig:checkpointsBtwo}); therefore we decelerate for $k_1$ steps, and then keep the same velocity. The second checkpoint is then given by the unique integer value that {\itshape aligns} with the $\drefmax$ trajectory. This can be straightforwardly computed by $k_2 = \ell - s + s' + k_1$;
\item $k_1 \in \mathbb{R}_{>0}$, then we take $\lfloor k_1 \rfloor$ (right Figure \ref{fig:checkpointsBtwo}); similarly we have $A(\lfloor k_1 \rfloor) < D < A(\lfloor k_1 \rfloor + 1)$; this time, we need to remove an extra area portion by doing a 1 step deceleration in a certain integer point between $\lfloor k_1 \rfloor$ and $\ell-s'+s+\lfloor k_1 \rfloor$. This is given by  $k_2 = D - (\lfloor k_1 \rfloor)(\ell-s+s')$. 
\end{itemize}

Hence we can describe such a trajectory in a compact and deterministic way by a sequence of pairs $(a,b)$:
\begin{enumerate}
	\item In the first case, we have $(-,k_1), (0,\ell-s+s'), (-, s-s'-k_1)$
	\item Otherwise, we have $(-,\lfloor k_1 \rfloor), (0,\ell-s+s'-k_2), (-,1), (0, k_2), (-, s-s'-\lfloor k_1 \rfloor-1)$
\end{enumerate}

\noindent{\bf Case 3 ($\delta \in$ {\bf C})}:

\begin{figure}[!htb]
	\centering
	\begin{subfigure}[!t]{0.49\textwidth}
		\centering
		\begin{tikzpicture}[thick, scale=0.6, every node/.style={scale=1}]
			\coordinate (s) at (0,1);
			\coordinate (e) at (7,2);
			\coordinate (dmax) at (4,5);

			\draw (-0.15,2) node[left]{$s'$}--(0.15,2);
			\draw (-0.15,1) node[left]{$s$}--(0.15,1);
			\draw (7,-0.15) node[below]{$\ell$}--(7,0.15);

			\draw[-stealth](0,-1) -- (0,6) node[thick, scale=0.9,above] {speed};
			\draw[-stealth](0,0) -- (7.5,0) node[thick, scale=0.9,above] {length};

			\draw (2.5,0.15)--(2.5,-0.15) node[below, scale=0.6]{$k_1$};
			\draw (5.5,0.15)--(5.5,-0.15) node[below, scale=0.6]{$k_2$};

			\draw [draw=none, fill=gray, opacity=0.3] (1,2) -- (2.5,3.5) -- (5.5,3.5) -- (e) -- cycle;


			\draw [color=red, line width=0.3mm] (2.5,3.5) -- (dmax) -- (5.5,3.5);
			\draw [color=orange, line width=0.3mm] (2.5,3.5) -- (5.5,3.5);
			\draw [color=brown, line width=0.3mm] (1,2) -- (e);
			
			\draw [color=brown, dashed, line width=0.3mm, dash pattern=on 2pt off 2pt,dash phase=2pt] (s) -- (1,2);
			\draw [color=orange, dashed,line width=0.3mm, dash pattern=on 2pt off 2pt] (s) -- (1,2);
			
			\draw [color=red, dashed, line width=0.3mm, dash pattern=on 2pt off 2pt,dash phase=2pt] (1,2) -- (2.5,3.5);
			\draw [color=orange, dashed,line width=0.3mm, dash pattern=on 2pt off 2pt] (1,2) -- (2.5,3.5);

			\draw [color=red, dashed, line width=0.3mm, dash pattern=on 2pt off 2pt,dash phase=2pt] (5.5,3.5) -- (e);
			\draw [color=orange, dashed,line width=0.3mm, dash pattern=on 2pt off 2pt] (5.5,3.5) -- (e);

		\end{tikzpicture}
	\end{subfigure}
	\begin{subfigure}[!t]{0.49\textwidth}
		\centering
		\begin{tikzpicture}[thick, scale=0.6, every node/.style={scale=1}]
			\coordinate (s) at (0,1);
			\coordinate (e) at (7,2);
			\coordinate (dmax) at (4,5);

			\draw (-0.15,2) node[left]{$s'$}--(0.15,2);
			\draw (-0.15,1) node[left]{$s$}--(0.15,1);
			\draw (7,-0.15) node[below]{$\ell$}--(7,0.15);

			\draw[-stealth](0,-1) -- (0,6) node[thick, scale=0.9,above] {speed};
			\draw[-stealth](0,0) -- (7.5,0) node[thick, scale=0.9,above] {length};

			\draw (2.5,0.15)--(2.5,-0.15) node[below, scale=0.6]{$k_1$};
			\draw (4,0.15)--(4,-0.15) node[below, scale=0.6]{$k_2$};

			\draw [draw=none, fill=gray, opacity=0.3] (1,2) -- (2.5,3.5) -- (4,3.5) -- (4.5,3) -- (6,3) -- (e) -- cycle;


			\draw [color=red, line width=0.3mm] (2.5,3.5) -- (dmax) -- (6,3);
			\draw [color=orange, line width=0.3mm] (2.5,3.5) -- (4,3.5) -- (4.5,3) -- (6,3);
			\draw [color=brown, line width=0.3mm] (1,2) -- (e);
			
			\draw [color=brown, dashed, line width=0.3mm, dash pattern=on 2pt off 2pt,dash phase=2pt] (s) -- (1,2);
			\draw [color=orange, dashed,line width=0.3mm, dash pattern=on 2pt off 2pt] (s) -- (1,2);
			
			\draw [color=red, dashed, line width=0.3mm, dash pattern=on 2pt off 2pt,dash phase=2pt] (1,2) -- (2.5,3.5);
			\draw [color=orange, dashed,line width=0.3mm, dash pattern=on 2pt off 2pt] (1,2) -- (2.5,3.5);

			\draw [color=red, dashed, line width=0.3mm, dash pattern=on 2pt off 2pt,dash phase=2pt] (6,3) -- (e);
			\draw [color=orange, dashed,line width=0.3mm, dash pattern=on 2pt off 2pt] (6,3) -- (e);

		\end{tikzpicture}
	\end{subfigure}
	\caption{Two possible trajectories of distance $\delta$ in the {\bf C} region. The left scenario happens whenever $A(k)$ matches exactly $D$, while the right scenario whenever $A(k)$ is an \emph{overestimation} of $D$}
	\label{fig:checkpointsC}
\end{figure}

We can express the \emph{true} integer distance between $D = \delta-\drefmax(\ell)$. Similarly as Case 1, $A(k)$ satisfies the same property.

We solve the equation $A(k)=D$. By taking the smaller root of this equation $k_1$, we distinguish between two cases:
\begin{itemize}
\item $k_1 \in \mathbb{Z}_{\ge 0}$: if $k_1$ is integer, then $A(k)$ matches exactly $D$ (left Figure \ref{fig:checkpointsC}); therefore, the equation $A(k)$ above coincides exactly with the amount of area to add to $\drefmax$ in order to 
match $\delta$.\\
The second checkpoint is then given by the unique integer value that {\itshape aligns} with the $\dmax$ trajectory. This can be straightforwardly computed as $k_2 = \ell - |s'-s| - k_1$;
\item $k_1 \in \mathbb{R}_{> 0}$: if $k_1$ is real, then $A(k)$ does not match exactly $D$, and we take $\lceil k_1 \rceil$ (right Figure \ref{fig:checkpointsC}); This implies that $A(\lceil k_1 \rceil) - D > 0$. Furthermore, by construction, $A(\lceil k_1 \rceil-1) < D < A(\lceil k_1 \rceil)$,
hence the $A(k)$ construction has an excess area $A(\lceil k_1\rceil)-D$. Therefore, there exists another integer value in the range $\left[ \lceil k_1 \rceil, \ell-|s'-s|- \lceil k1 \rceil \right]$ where we need 
to \emph{decrease} the area of $A(k)$ in order to match exactly $D$. This required adjustment is given by $k_2 = A(\lceil k_1 \rceil)-D$, which is a valid integer value since both $A(\lceil k_1 \rceil)$ and $D$ are integers. 
\end{itemize}

We can describe such a trajectory in a compact way by a sequence of pairs $(a,b)$.
\begin{itemize} 
	\item When $s\leq s'$:
	\begin{enumerate}
		\item In the first case (if $A(k)$ is exactly $D$), we have $(+,s'-s+k_1), (0,\ell-s'+s-2k_1), (-, k_1)$
		\item Otherwise, we have $(+,s'-s+\lceil k_1 \rceil), (0,\ell-s'+s-2\lceil k_1 \rceil-k_2), (-,1), (0,k_2), (-, \lceil k_1 \rceil-1)$
	\end{enumerate}
	\item When $s>s'$:
	\begin{enumerate}
		\item In the first case (if $A(k)$ is exactly $D$), we have $(+,k_1), (0,\ell-s+s'-2k_1), (-, s-s'+k_1)$
		\item Otherwise, we have $(+,\lceil k_1 \rceil), (0,\ell-s+s'-2\lceil k_1 \rceil-k_2), (-,1), (0,k_2), (-, s-s'+\lceil k_1 \rceil-1)$
	\end{enumerate}
\end{itemize} 

To verify the constructions, write $n=\ell-|s'-s|$. In Cases~1 and~3, the displacement bounds give $0\leq D\leq\lfloor n^2/4\rfloor$. Thus the integer layer count $q=\lceil k_1\rceil$ satisfies $0\leq q\leq\lfloor n/2\rfloor$. When the root is not integral, $q\geq1$ and the residual $r=A(q)-D$ satisfies
$0<r<A(q)-A(q-1)=n-2q+1$, hence $r\leq n-2q$. In Case~2, $n>0$ and $0\leq D\leq n|s'-s|$; writing $D=nq+r$ gives $0\leq r<n$ and $q\leq |s'-s|$, with $q<|s'-s|$ whenever $r>0$. These bounds show that every displayed run duration is a nonnegative integer.

In each case, the durations sum to $\ell$ and the signed acceleration increments sum to $s'-s$. A run $(a,m)$ starting at speed $v$ contributes displacement $mv+a m(m+1)/2$. Summing these contributions gives $\drefmin(\ell)-D$ in Case~1, $\drefmin(\ell)+D$ or $\drefmax(\ell)-D$ in the two subcases of Case~2, and $\drefmax(\ell)+D$ in Case~3, respectively, all equal to $\delta$. Thus each construction is a valid trajectory from $c$ to $c'$ of length $\ell$. The boundary case $\ell=|s'-s|$ is covered by the first remark. Each representation uses at most five acceleration runs, with parameters obtained by a constant number of arithmetic, square-root, and rounding operations, proving the theorem.

\end{document}